\documentclass[12pt,onecolumn,oneside,a4paper,peerreview]{IEEEtran_kit}
\usepackage{cite}
\usepackage{float}
\usepackage{graphicx}
\usepackage{amsmath}
\usepackage{times}
\usepackage{graphicx}
\usepackage{latexsym}
\usepackage{graphicx}
\usepackage{bm}
\usepackage{amssymb}
\usepackage[center]{caption2}
\usepackage{stfloats}
\usepackage{cases}
\usepackage{array}
\usepackage{setspace}
\usepackage{fancyhdr}
\usepackage{color}
\usepackage{subfigure}
\usepackage{citesort}
\usepackage{multirow}
\usepackage{amsmath}
\usepackage{cases}
\usepackage[noend]{algpseudocode}
\usepackage{algorithmicx,algorithm}
\def\bm#1{\mbox{\boldmath $#1$}}

\newtheorem{theorem}{Theorem}
\newtheorem{lemma}{Lemma}

\newtheorem{corollary}{Corollary}

\newtheorem{proposition}{Proposition}
\newtheorem{remark}{Remark}
\renewcommand{\baselinestretch}{0.97}

\begin{document}
\title{Angle-Domain Intelligent Reflecting Surface Systems: Design and Analysis}
\author{\authorblockN{Xiaoling Hu,  Caijun Zhong, and Zhaoyang Zhang
\thanks{X. Hu, C. Zhong, and Z. Zhang are with the College of Information Science and Electronic Engineering, Zhejiang University, Hangzhou, China (Email:  \{11631052, caijunzhong, ning\_ming\}@zju.edu.cn).}
}}
\maketitle
\begin{abstract}
This paper considers an angle-domain intelligent reflecting surface (IRS)  system. We derive maximum likelihood (ML) estimators for the effective angles from the base station (BS) to the user  and  the effective angles of propagation from the IRS to the user. It is demonstrated that the accuracy of the estimated angles improves with the number of BS antennas. Also, deploying the IRS closer to the BS increases the  accuracy of the estimated angle from the IRS to the user.
Then, based on the estimated angles, we propose a joint optimization of BS beamforming and IRS beamforming, which achieves  similar performance to two   benchmark algorithms based on full CSI  and the multiple signal classification (MUSIC) method respectively. Simulation results show that the optimized BS beam  becomes more focused towards the IRS direction as the number of reflecting elements increases.
Furthermore, we derive a closed-form  approximation, upper bound  and lower bound for the achievable rate. The analytical findings indicate that  the achievable rate can be  improved by increasing the number of BS antennas or reflecting elements. Specifically, the BS-user link and  the BS-IRS-user link can obtain power gains of order $N$ and $NM^2$, respectively, where $N$ is the antenna number and $M$ is the number of reflecting elements.
\end{abstract}

\newpage

\section{Introduction}
Wireless networks are anticipated to connect over 28.5
billion devices by the year 2022\cite{WhitePaper}, giving rise to questions of coverage, quality of service, reliability and energy consumption.
Recently, the intelligent reflecting surface (IRS) has been proposed as  a promising technology to address these questions, due to its features of both low energy consumption\cite{huang2019reconfigurable,Hu1,Hu2,Hu3,Hu4,wu2019weighted,bjornson2019intelligent} and low hardware complexity\cite{han2019large,qingqing2019towards}.

Specifically, an IRS is a meta-surface equipped with a large number of nearly passive, low-cost, reflecting elements with programmable parameters\cite{cui2014coding,liaskos2018new}. Assisted by a smart controller, each reflecting element independently reflects the incident  signal with adjustable phase shifts. With the ability to smartly alter phase shifts, the IRS can generate directional beams which are dynamically adjusted according to the operating requirements and conditions, thereby realizing programmable propagation environments \cite{chen2019intelligent}. Given that there is no amplifier used, the IRS is more energy-efficient than an  Amplify-and-Forward (AF) relay. Another benefit of the IRS is that it can potentially be easily deployed, for example on  building facades, due to its small size that is the consequence of an absence of radio frequency (RF) circuitry \cite{huang2018energy}.

Hence, IRS-aided communications have attracted considerable attention as a promising means to realize dynamically adjustable wireless propagation channels, which can be used e.g. to fill coverage holes at a low cost. In particular, there have been a plethora of works on IRS beamforming, see e.g.  \cite{wu2019intelligent,wu2018intelligent,ye2019joint,guo2019weighted,huang2018energy,yang2019intelligent,wang2019intelligent,pan2019intelligent,chu2019intelligent,wu2019beamforming,abeywickrama2019intelligent,yu2020irs}. In \cite{wu2019intelligent} and \cite{wu2018intelligent}, the authors first presented the joint optimization of IRS beamforming and BS transmit beamforming. By applying a semidefinite relaxation (SDR) method, the total transmit power at the base station (BS) is minimized. Later on,  \cite{ye2019joint} considered the minimization  of the symbol error rate by iteratively optimizing the phase shifts and BS precoder, while \cite{guo2019weighted} focused on
maximizing the weighted sum-rate by adopting a Lagrangian dual transform  technique under the assumption of discrete phase shifts.
Finally, in  \cite{huang2018energy}, the energy efficiency maximization problem  was solved via gradient descent search and sequential fractional programming.
Besides, the integration of IRS with other promising technologies has also been investigated, for instance, non-orthogonal multiple access (NOMA) \cite{yang2019intelligent}, millimeter wave (mmWave) \cite{wang2019intelligent}, simultaneous wireless information and power transfer (SWIPT) \cite{pan2019intelligent},   cognitive radio \cite{cognitive}, two-way relaying \cite{Zhangyu}, and physical layer security \cite{Xuwei}.

All the prior works on IRS beamforming  assumed that  perfect channel state information (CSI) is available  at both the BS and the IRS. How to obtain this CSI, however, is still a challenging issue for an IRS system due to the typically large number of reflecting elements and the passive architecture of the IRS.
The typical channel estimation scheme is to estimate the cascaded IRS-aided channel by using various reflection patterns to train all the reflecting elements \cite{nadeem2019intelligent,zheng2019intelligent,he2019cascaded,mishra2019channel}. For example, an element-by-element ON/OFF-based reflection pattern was adopted in \cite{nadeem2019intelligent}, while a new IRS reflection (phase-shift) pattern satisfying the unit-modulus constraint was proposed in \cite{zheng2019intelligent}.
However, the channel training overhead of such a scheme would become prohibitively high as the number of reflecting elements increases.
Besides, since the  CSI is estimated only at the transmitter and receiver,  the IRS has no access to the CSI needed to perform beamforming.
To tackle this problem, the works \cite{subrt2012intelligent,wu2018intelligent}  proposed a two-mode IRS model, where the IRS controller can switch between receiving mode for CSIs and reflecting mode for data transmission.
However, the realization of a receiving mode would lead to higher hardware complexity as well as more power consumption, and diminishes the main strengths of the IRS, namely its low complexity and cost.
  To  reduce the hardware cost, a semi-passive IRS architecture has been proposed \cite{estimation_ML,estimation_CS}, where only a small proportion of IRS elements have the capability of both sensing and reflecting (refereed to as semi-passive elements). As such, with the aid of these semi-passive elements, the IRS estimates the channels between itself and BS/users by leveraging on machine learning \cite{estimation_ML} or compressive sensing \cite{estimation_CS}. However, the computational complexity of the machine learning and compressive sensing methods is usually very high.

Motivated by these observations, in this paper, we propose an angle-domain IRS-aided system model, where effective angles are estimated  by using a low-complexity method, and then exploited to design the BS beam and IRS phase shifts. \footnote{  Like the work \cite{xu2020resource}, the proposed angle-domain framework  exploits the angle sparsity. As such, the channel estimation and the corresponding beamforming design are more efficient compared to those of the conventional IRS-assisted communication. However, for non-sparse channels where there are a large number of scatters in the radio frequency environment, the proposed angle-domain design framework is no longer suitable, since the angular information detection may also be complicated and time-consuming.}
  Although  the angle-domain signal processing technique has been well studied in the existing works (e.g.,\cite{xu2020resource,jian2019angle,shao2020angle,dong2020high}),  this is the first time to introduce the concept of angle domain to  deal with the channel estimation problem in the IRS-aided communication system. By exploiting angle sparsity, the training overhead as well as the computational complexity can be significantly reduced.
The main advantages of this model are listed below:
\begin{itemize}
    \item Only a very limited number of angles, instead of a large-scale channel matrix corresponding to the IRS, are needed, thereby avoiding massive channel training overhead and high computational complexity.
    \item Due to the above, only a very small amount of CSI needs to be communicated from the BS to the IRS. This allows the passive IRS to obtain CSI without adding too much hardware complexity, by linking it to the BS via a low-capacity hardware line.  As such, the BS and the IRS can exchange information (CSI and phase shifts), and thus the joint optimization of BS beamforming and IRS beamforming can be realized.
\end{itemize}
The contributions of this paper are summarized as follows:
\begin{itemize}
    \item  An angle-based channel estimation scheme with extremely low complexity has been proposed. In particular, the maximum likelihood (ML) estimators for effective angles from the BS to the user are presented, based on which the  effective angles from the IRS to the user are derived.
    \item  With a Gaussian model for the estimation error, we propose a joint optimization algorithm for BS beamforming and IRS beamforming,  which achieves  similar performance to two   benchmark algorithms based on full CSI  and the multiple signal classification (MUSIC) method respectively.  Furthermore, we have proved mathematically that the average received signal power is non-decreasing after each iteration of the proposed algorithm.
    \item The expression for the achievable rate is presented, based on which a closed-form approximation, an upper bound   for the achievable rate are derived. The analytical findings show that the BS-user link can obtain a power gain proportional to the number of antennas $N$, while the BS-IRS-user link can obtain a power gain of order $NM^2$.
\end{itemize}

The remainder of the paper is organized as follows. In Section \ref{s1}, we introduce the angle-domain IRS-based system, while in Section \ref{s2}, we propose an angle estimation scheme. Based on the estimated  effective angles, a  joint optimization algorithm for BS beamforming and IRS beamforming is presented in Section \ref{s3}. Then, we give a detailed analysis of the achievable rate in Section \ref{s4}. Numerical results and discussions are provided in Section \ref{s5}, and finally Section \ref{s6} concludes the paper.

Notation: Boldface lower case and upper case letters are used for column vectors and matrices, respectively. The superscripts ${\left(\cdot\right)}^{*}$, ${\left(\cdot\right)}^{T}$, ${\left(\cdot\right)}^{H}$, and ${\left(\cdot\right)}^{-1}$ stand for the conjugate, transpose, conjugate-transpose, and matrix inverse, respectively. Also, the Euclidean norm, absolute value, Hadamard product  are denoted by $\left\| \cdot \right\|$, $\left|\cdot\right|$ and $\odot$ respectively. In addition, $\mathbb{E}\left\{\cdot\right\}$ is the expectation operator, and $\text{tr}\left(\cdot\right)$ represents the trace. $\text{Re}\left( \cdot\right)$ denotes the real part. For a matrix ${\bf A}$, ${[\bf A]}_{mn}$ denotes its entry in the $m$-th row and $n$-th column. For real numbers $a$ and $b$, $\text{Remainder}(\frac{a}{b})$ and $\text{Quotient}(\frac{a}{b})$ respectively denote the Remainder and the quotient of $a$ divided by $b$. Besides, $j$ in $e^{j \theta}$ denotes the imaginary unit.
Finally, $z \sim \mathcal{CN}(0,{\sigma}^{2})$ denotes a circularly symmetric complex Gaussian random variable (RV) $z$ with zero mean and variance $\sigma^2$, and $z \sim \mathcal{N}(0,{\sigma}^{2})$ denotes a real valued Gaussian RV.

\section{system model} \label{s1}
We consider an angle-domain IRS aided system, as illustrated in Fig.\ref{f0}, where the BS equipped with $N$ antennas communicates with a single-antenna user assisted by  an IRS with $M$ reflecting elements. Both the BS and the IRS have uniform rectangular arrays (URA) with the sizes of $\sqrt{N} d_\text{BS}\times \sqrt{N}d_\text{BS}$ and $\sqrt{M}d_\text{IRS}
\times\sqrt{M} d_\text{IRS}$ respectively, where   $d_{\text{BS}}$ is the  distance between two adjacent BS antennas and $d_{\text{IRS}}$ is the  distance between two adjacent reflecting elements.
Furthermore,
the BS is connected with the IRS controller via a backhaul link so that they can exchange information (e.g. CSI and phase shifts).

\begin{figure}[!ht]
  \centering
  \includegraphics[width=3in]{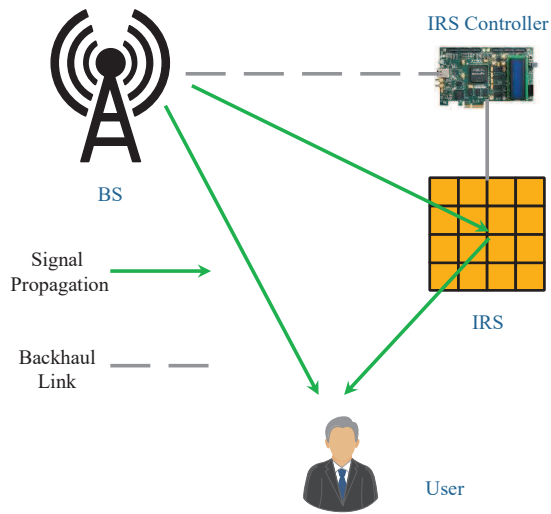}
  \caption{Model of the angle-domain IRS aided system }
  \label{f0}
\end{figure}

An angle-domain  Rician fading channel model is assumed, to model a channel with a line-of-sight (LOS) component. Such a channel corresponds to the practically desirable deployment of an IRS with LOS to both the BS and the user.
Specifically, the channel from the BS to the $\text{IRS}$ can be expressed as
\begin{align} 
{\bf H}_{\text{B2I}}=\sqrt{\alpha_{\text{B2I}}\frac{v_{\text{B2I}}}{v_{\text{B2I}}+1}} {\bf b} \left(\bar\theta_{\text{x-B2Ia}},\bar\theta_{\text{y-B2Ia}} \right)
{\bf a}^T \left(\bar \theta_{\text{x-B2I}},\bar\theta_{\text{y-B2I}} \right)+\sqrt{\alpha_{\text{B2I}}\frac{1}{v_{\text{B2I}}+1}} \widetilde{\bf H}_{\text{B2I}},
\end{align}
where $\widetilde{\bf H}_{\text{B2I}}$ is the non-line of sight (NLOS) component, whose elements have the $\mathcal{CN}\left(0,1\right)$ distribution, $\alpha_{\text{B2I}}$ models large-scale fading, $v_{\text{B2I}}$ is the Rician K-factor, ${\bf a}$ and ${\bf b}$ are array steering vectors of the BS and the IRS, respectively. The two  effective angles  of departure (AODs), i.e., the phase differences between two adjacent antennas  along the $x$ and $y$ axes,  are respectively given by
 \begin{align}
 \bar\theta_{\text{x-B2I}}=-\frac{2 \pi d_{\text{BS}}}{\lambda} \cos \theta_{\text{B2I}} \cos \phi_{\text{B2I}},\\
 \bar\theta_{\text{y-B2I}}=-\frac{2 \pi d_{\text{BS}}}{\lambda} \cos \theta_{\text{B2I}} \sin \phi_{\text{B2I}},
  \end{align}
 where
$\lambda$ is the carrier wavelength, $ \theta_{\text{B2I}}$ and $\phi_{\text{B2I}}$ are the elevation and azimuth  AODs from the BS to the IRS, respectively. Similarly, the two  effective angles of arrival (AOAs) are
\begin{align}
 \bar\theta_{\text{x-B2Ia}}=\frac{2 \pi d_{\text{IRS}}}{\lambda} \cos \theta_{\text{B2Ia}} \cos \phi_{\text{B2Ia}},\\
 \bar\theta_{\text{y-B2Ia}}= \frac{2 \pi d_{\text{IRS}} }{\lambda} \cos \theta_{\text{B2Ia}} \sin \phi_{\text{B2Ia}},
 \end{align}
 where $ \theta_{\text{B2Ia}}$ and $\phi_{\text{B2Ia}}$ are the elevation and azimuth AOAs at the IRS, respectively. Furthermore, we assume $d_{\text{BS}}=d_{\text{IRS}}=\frac{\lambda}{2}$.

Specifically,  the $n$-th element of ${\bf a} \left( \bar\theta_{\text {x-B2I}}, \bar\theta_{\text {y-B2I}} \right)$ and the $m$-th element of ${\bf b} \left( \bar\theta_{\text {x-B2I}}, \bar\theta_{\text {y-B2I}} \right)$ are given by
\begin{align}
&a_n=e^{j\left(
 i_{N,n}  \bar\theta_{\text {x-B2I}}+j_{N,n}  \bar\theta_{\text {y-B2I}}\right)}, \\
&b_{m}=e^{ j\left(
i_{M,m} \bar\theta_{\text {x-B2I}} + j_{M,m}  \bar\theta_{\text {y-B2I}}  \right)},
\end{align}
where
\begin{align}
&  i_{N,n}= \text{Remainder}\left( \frac{n-1}{\sqrt{N}} \right),n=1,...,N,\\ &j_{N,n}=\text{Quotient}\left( \frac{n-1}{\sqrt{N}} \right),n=1,...,N,
\end{align}
\begin{align}
 & i_{M,m}= \text{Remainder}\left( \frac{m-1}{\sqrt{M}} \right),m=1,...,M,\\ &j_{M,m}=\text{Quotient}\left( \frac{m-1}{\sqrt{M}} \right),m=1,...M.
\end{align}

Similarly, the channel from the IRS to the  user is given by
\begin{align}
{\bf h}_{\text{I2U}}^T=\sqrt{\alpha_{\text{I2U}}\frac{v_{\text{I2U}}}{v_{\text{I2U}}+1}}
{\bf b}^T \left(\bar\theta_{\text{x-I2U}},\bar\theta_{\text{y-I2U}} \right)+\sqrt{\alpha_{\text{I2U}}\frac{1}{v_{\text{I2U}}+1}} \widetilde{\bf h}_{\text{I2U}}^T .
\end{align}

Also, the channel from the BS to the  user  can be expressed as
\begin{align}
{\bf h}_{\text{B2U}}^T=\sqrt{\alpha_{\text{B2U}}\frac{v_{\text{B2U}}}{v_{\text{B2U}}+1}}
{\bf a}^T \left(\bar\theta_{\text{x-B2U}},\bar\theta_{\text{y-B2U}} \right)+\sqrt{\alpha_{\text{B2U}}\frac{1}{v_{\text{B2U}}+1}} \widetilde{\bf h}_{\text{B2U}}^T.
\end{align}

During the data transmission, the BS transmits the signal ${\bf w}s$, where ${\bf w}$ is the transmitting beam and $s$ is the data intended for the user. Then, the signal received by the user is given by
\begin{align}
y_{\text{U}}=\left( {\bf h}_{\text{B2U}}^T  + {\bf h}_{\text{I2U}}^T {\bm \Theta} {\bf H}_{\text{B2I}} \right)  {\bf w} s+n_{\text{U}},
\end{align}
where ${\bf w}$ is the BS transmitting beam,  $ n_{\text{U}} \sim \mathcal{CN}\left( 0, \sigma_0^2\right)$ is the noise at the user, and the phase shift matrix ${\bm \Theta}$ is given by  ${\bm \Theta}= {\text {diag}} \left( {\bm \xi}\right)$
with the phase shift  beam $ {\bm \xi}={[ e^{j\vartheta_1},..., e^{j\vartheta_n},...,  e^{j\vartheta_M}]}^T$.

To facilitate the design of BS transmitting beam ${\bf w}$ and IRS phase shift beam ${\bm \xi}$, we first perform angle estimation as illustrated in the next section.

%


\section{Angular Information detection Method}\label{s2}
In this section, we propose a detection method to estimate the  effective angles  from the BS to the user ($\bar\theta_{\text{x-B2U}}$ and $\bar\theta_{\text{y-B2U}}$), based on  which the  effective angles from the IRS to the user
($\bar\theta_{\text{x-I2U}}$ and $\bar\theta_{\text{y-I2U}}$)  are derived. Later, these effective angles will be exploited to design the BS beam and IRS phase shifts in Section \ref{s3}.

During the angle estimation period, the IRS is turned off. The user transmits an unmodulated carrier at the power $P_q$ to the BS. The received baseband signal at the BS is given by
\begin{align}
{\bf r}={\bf h}_{\text{B2U}}^* q+{\bf n}_{\text{BS}}=
\sqrt{\frac{\alpha_{\text{B2U}} v_{\text{B2U}}}{v_{\text{B2U}}+1}}
{\bf a} \left(-\bar\theta_{\text {x-B2U}},-\bar\theta_{\text {y-B2U}} \right)  q +\sqrt{\frac{\alpha_{\text{B2U}} }{v_{\text{B2U}}+1}} \widetilde{\bf h}^*_\text{B2U}q+{\bf n}_{\text{BS}},
\end{align}
where $ q = \sqrt{P_q}e^{j\theta_q}$ is the baseband-equivalent representation of the unmodulated carrier, and ${\bf n}_{\text{BS}}\sim \mathcal{CN} \left( 0,\sigma_{\text{BS},0}^{2}\right) $ is the noise at the BS.

In particular, the signal received by the $n$-th antenna  is
\begin{align}
r_n&=\sqrt{\frac{\alpha_{\text{B2U}} v_{\text{B2U}}}{v_{\text{B2U}}+1}} a_n q + \sqrt{\frac{\alpha_{\text{B2U}} }{v_{\text{B2U}}+1}}\tilde{h}_{\text{B2U},n} q +n_{\text{BS},n} \\
&=\sqrt{\frac{\alpha_{\text{U}} v_{\text{B2U}} }{v_{\text{B2U}}+1}} e^{\left(
 -i_{N,n}  \bar\theta_{\text {x-B2U}}-j_{N,n}  \bar\theta_{\text {y-B2U}}\right)} q+\sqrt{\frac{\alpha_{\text{B2U}} }{v_{\text{B2U}}+1}}\tilde{h}_{\text{B2U},n} q +n_{\text{BS},n}.\nonumber
\end{align}

The phase of $r_n$ can be decomposed as
\begin{align}
\vartheta_{n}= \theta_{\text{q}}- i_{N,n}\bar\theta_{\text {x-B2U}}-j_{N,n}  \bar\theta_{\text {y-B2U}} +e_{n},
\end{align}
where $e_{n}$ is the phase uncertainty caused by noise and NLOS paths.

\begin{proposition} \label{p1}
With a large Rician K-factor and a high received signal-to-noise ratio (SNR) at the BS, the phase uncertainty $e_n$  is approximately circularly symmetric Gaussian $\mathcal{CN}\left(0, \sigma_e^2\right)$ with
\begin{align}
\sigma_e^2= \frac{ 4-\pi }{8 v_{\text{B2U}} }  +\frac{\left( 4-\pi \right)\left(v_{\text{B2U}}+1\right)}{8 \alpha_{\text{B2U}} P_q v_{\text{B2U}} }  \sigma_{\text{BS},0}^{2} .
\end{align}
\end{proposition}

\begin{proof}
See Appendix \ref{A1}.
\end{proof}
\begin{remark}
From Proposition \ref{p1}, we can observe that the phase uncertainty is a decreasing function with respect to the transmit power. Also, a larger Rician K-factor would lead to less phase uncertainy.
\end{remark}

The observed phase difference between $r_n$ and $r_m$ is
\begin{align}\label{E20}
\Delta \bar\theta_{n,m}=r_n-r_m=-\left(i_{N,n}-i_{N,m}\right) \bar\theta_{\text {x-B2U}}-\left(j_{N,n}-j_{N,m}\right)  \bar\theta_{\text {y-B2U}} +\Delta e_{n,m},
\end{align}
where $\Delta e_{n,m}=e_{n}-e_{m}\sim \mathcal{CN} \left(0, \sigma_{\text{pd}}^2 \right)$ with $ \sigma_{\text{pd
}}^2 = 2 \sigma_{e}^2$.

Then, we use the observed phase difference $\Delta \bar\theta_{n,m}$ to estimate $\bar\theta_{\text {x-B2U}}$ and $\bar\theta_{\text {y-B2U}  }$. Specifically, we choose $\frac{N}{2}$ phase differences
$\Delta \bar\theta_{n,m_n},n=1,...,\frac
{N}{2}, m_n=N-n+1$.

\begin{theorem} {\label{t1}}
The ML estimators for $\bar\theta_{\text {x-B2U}}$ and $\bar\theta_{\text {y-B2U}}$ are given by
\begin{align}
\hat{\bar\theta}_{\text {x-B2U}}=-\frac{6 \sum\limits_{n=1}^{N/2}\left(i_{N,n}-i_{N,m_n}\right) \Delta \bar\theta_{n,m_n}}
{N\left(N-1\right)},\\
\hat{\bar\theta}_{\text {y-B2U}}=-\frac{6 \sum\limits_{n=1}^{N/2}\left(j_{N,n}-j_{N,m_n}\right) \Delta \bar\theta_{n,m_n}}
{N\left(N-1\right)}.
\end{align}
\end{theorem}
\begin{proof}
See Appendix \ref{A2}.
\end{proof}

\begin{corollary} \label{c1}
 The  effective angles ${\bar\theta}_{\text {x-B2U}}$ and ${\bar\theta}_{\text {y-B2U}}$ can be decomposed as
 \begin{align}
\bar\theta_{\text{x-B2U}}=\hat{\bar\theta}_{\text{x-B2U}}+\epsilon_{\text{x-B2U}},\\ \bar\theta_{\text{y-B2U}}=\hat{\bar\theta}_{\text{y-B2U}} +\epsilon_{\text{y-B2U}},
\end{align}
where  $\epsilon_{\text{x-B2U}}$ and $\epsilon_{\text{y-B2U}}$ are estimation errors, which follow Guassion distribution $\mathcal{N}\left(0,\sigma_{\text{est}}^2\right)$ with the variance
\begin{align} \label{delta_est}
\sigma_{\text{est}}^2=\frac{6\sigma_{\text{pd}}^2 }{N\left(N-1\right)}.
\end{align}
\end{corollary}
\begin{proof}
Starting from Theorem \ref{t1},
and noticing that
\begin{align}
\Delta \bar\theta_{n,m_n} \sim \mathcal{CN}\left(-\left( i_{N,n}-i_{N,m_n}\right)\bar\theta_{\text {x-B2U}}-\left( j_{N,n}-j_{N,m_n}\right)  \bar\theta_{\text {y-B2U}}, \sigma_{\text{pd}}^2 \right),
\end{align}
we can obtain the desired result.
\end{proof}

\begin{remark}
Corollary \ref{c2} shows that the variance of the estimation error drops approximately inversely with the square of the  antenna number, which indicates the great benefit of applying a large number of BS antennas in angle estimation.
\end{remark}

With a LOS path between the BS and the user, their distance  $d_{\text{B2U}}$ can be obtained with high accuracy, by measuring time difference of arrival, i.e., wave propagation time from the user to the BS \cite{simic2011positioning}.

After obtaining the effective angles and the distance, we can estimate the location of the user
as $(\hat{x}_{\text{U}},\hat{y}_{\text{U}},\hat{z}_{\text{U}})$ with
\begin{align}
&\hat{x}_{\text{U}}=-\frac{ d_{\text{B2U}} \hat{\bar \theta}_{\text{x-B2U}} }{\pi},\\
&\hat{y}_{\text{U}}=-\frac{d_{\text{B2U}} \hat{\bar \theta}_{\text{y-B2U}} }{\pi},\\
&\hat{z}_{\text{U}}=-\frac{ d_{\text{B2U}}  \sqrt{ \pi^2- \hat{\bar \theta}_{\text{x-B2U}}^2- \hat{\bar \theta}_{\text{y-B2U}}^2   }  }{\pi},
\end{align}
where we assume the BS is located at $(0,0,0)$.

With the estimated user location  $(\hat{x}_{\text{U}},\hat{y}_{\text{U}},\hat{z}_{\text{U}})$  and the perfectly known IRS location  $(x_{\text{I}},y_{\text{I}},z_{\text{I}})$, the BS can calculate the  effective angles from the IRS to the user as follows:

\begin{lemma}\label{L1}
The  effective angles from the IRS to the user are
\begin{align}
\hat{\bar \theta}_{\text{x-I2U}}=\frac{\left({x}_{\text{I}}-\hat x_{\text{U}} \right)\pi}{\hat{d}_{\text{I2U}}},\\
\hat{\bar \theta}_{\text{y-I2U}}=\frac{\left({y}_{\text{I}}-\hat y_{\text{U}} \right)\pi}{\hat{d}_{\text{I2U}}},
\end{align}
where
\begin{align}
\hat{d}_{\text{I2U}}=\sqrt{{\left(\hat{x}_{\text{U}}-x_{\text{I}} \right)}^2+{\left(\hat{y}_{\text{U}}-y_{\text{I}} \right)}^2+{\left(\hat{z}_{\text{U}}-z_{\text{I}} \right)}^2  }.
\end{align}

\end{lemma}

\begin{proof}
Lemma \ref{L1} is provable using elementary geometry.
\end{proof}

\begin{theorem} \label{t2}
 The  effective angles from the IRS to the user can be decomposed into
\begin{align}
&\bar\theta_{\text{x-I2U}}= \hat{\bar\theta}_{\text{x-I2U}}+ \varphi_{1}\epsilon_{\text{x-B2U}}+\varphi_{{2}}\epsilon_{\text{y-B2U}},\label{E3}\\
&\bar\theta_{\text{y-I2U}}=\hat{\bar\theta}_{\text{y-I2U}}+
\varphi_{2}\epsilon_{\text{x-B2U}}+\varphi_{3}\epsilon_{\text{y-B2U}},\label{E4}
\end{align}
with
\begin{align}
&\varphi_{1}= Ra  \left\{ 1-\frac{\hat{\bar \theta}_\text{x-I2U}^{2}}{\pi^2} +\frac{\hat{\bar \theta}_\text{x-I2U}^{2} \hat{\bar \theta}_\text{z-I2U}  }{\pi^3} \right\},\\
&\varphi_{2}= Ra   \left\{- \frac{\hat{\bar \theta}_{\text{x-I2U}} \hat{\bar \theta}_\text{y-I2U} }{\pi^2 }
+\frac{\hat{\bar \theta}_{\text{x-I2U}} \hat{\bar \theta}_\text{y-I2U}\hat{\bar \theta}_{\text{z-I2U}} }{\pi^3 }
\right\},\\
&\varphi_{3}= Ra \left\{ 1- \frac{\hat{\bar \theta}_\text{y-I2U}^{2}}{\pi^2}+ \frac{\hat{\bar \theta}_\text{y-I2U}^{2} \hat{\bar \theta}_\text{z-I2U}  }{\pi^3} \right\},
\end{align}
where $Ra \triangleq\frac{d_{\text{B2U}}}{\hat d_{\text{I2U}}}$
and $\hat{\bar \theta}_{\text{z-I2U}}\triangleq \frac{\left({z}_{\text{I}}-\hat z_{\text{U}} \right)\pi}{\hat{d}_{\text{I2U}}}$.
\end{theorem}

\begin{proof}
See Appendix \ref{A3}.
\end{proof}

\begin{remark}
Theorem \ref{t2} implies that in addition to the estimation error of $\hat{\bar \theta}_{\text{x-B2U}}$ and $\hat{\bar \theta}_{\text{y-B2U}}$, the accuracy of $\hat{\bar \theta}_{\text{x-I2U}}$ and $\hat{\bar \theta}_{\text{y-I2U}}$ is also significantly influenced by the distances of the IRS and the BS to the user. Specifically, the estimation error increases linearly with the ratio $Ra \triangleq\frac{d_{\text{B2U}}}{\hat d_{\text{I2U}}}$, implying that it is best to place the IRS close to the BS in terms of angle estimation accuracy. \footnote{ It is worth noting that, assuming perfect CSI, it is desirable to place the IRS close to the BS or the user \cite{qingqing2019towards,wu2019intelligent}. However, in this paper, imperfect CSI is considered, as such placing the IRS close to the user would not necessarily improve the achievable due to the increased angle estimation error.}
\end{remark}

\section{Joint optimization of BS beamforming and IRS beamforming }\label{s3}
Using the estimated angles, in this section we aim to maximize the average received signal power:
\begin{align}
P_{\text{r}} =\mathbb{E} \left\{ {\left|\left( {\bf h}_{\text{B2U}}^T  + {\bf h}_{\text{I2U}}^T {\bm \Theta} {\bf H}_{\text{B2I}} \right)  {\bf w}\right|}^2 \right\},
\end{align}
by jointly optimizing  BS beamforming  and IRS beamforming.

We first give the expression of  the average received signal power.
\begin{proposition} \label{p2}
The average received signal power is  given by
\begin{align}
P_{\text{r}}={\bf w}^H {\bf T} {\bf w},
\end{align}
with
\begin{align}
&{\bf T} \triangleq {\beta}_{\text{B2I2U}}
{\left( {\bm \Theta} \bar{\bf H}_{\text{B2I}}\right)}^H {\bf B}{\bm \Theta} \bar{\bf H}_{\text{B2I}}
+ \sqrt{{\beta}_{\text{B2I2U}} {\beta}_{\text{B2U}}} \left( {\left( {\bm \Theta} \bar{\bf H}_{\text{B2I}}\right)}^H {\bf C}+ {\bf C}^H  {\bm \Theta} \bar{\bf H}_{\text{B2I}} \right) \\
&+ {\beta}_{\text{B2U}} {\bf A}+\sigma_{\text{NLOS}}^2 {\bf I}_{N},\nonumber
\end{align}
where
\begin{align}
&\bar{\bf H}_\text{B2I}\triangleq {\bf b} \left(\bar\theta_{\text{x-B2Ia}},\bar\theta_{\text{y-B2Ia}} \right)
{\bf a}^T \left(\bar \theta_{\text{x-B2I}},\bar\theta_{\text{y-B2I}} \right)\label{d1}\\
&{\beta}_{\text{B2I2U}} \triangleq \frac{\alpha_{\text{I2U}} \alpha_{\text{B2I}} v_{\text{B2I}} v_{\text{I2U}}}{\left(v_{\text{B2I}}+1\right) \left(v_{\text{I2U}}+1\right)}, \label{d2}\\
&{\beta}_{\text{B2U}} \triangleq \frac{v_{\text{B2U}} \alpha_{\text{B2U}}   }{v_{\text{B2U}}+1},
\label{d3}\\
&\sigma_{\text{NLOS}}^2  \triangleq M \frac{\alpha_{\text{I2U}} \alpha_{\text{B2I}} }{v_{\text{B2I}}+1}
\left( 1+\frac{v_{\text{B2I}}}{v_{\text{I2U}}+1}\right)
+\frac{\alpha_{\text{B2U}} }{v_{\text{B2U}}+1}.\label{d4}
\end{align}
The elements of ${\bf A} \in \mathbb{C}^{N \times N}$, ${\bf B}\in \mathbb{C}^{M \times M}$ and ${\bf C} \in \mathbb{C}^{M \times N}$ are given by
\begin{align}
&{\left[{\bf A}\right]}_{mn}={\left[{\bf \hat A}\right]}_{mn}
\exp \left(- \frac{1}{2}\sigma_{\text{est}}^2\left\{ i_{N,mn}^2 +j_{N,mn} ^2\right\} \right), \\
&{\left[{\bf B}\right]}_{mn}={\left[{\bf \hat B}\right]}_{mn} \exp \left(- \frac{1}{2}\sigma_{\text{est}}^2\left\{ {\left(i_{M,mn}\varphi_1+ j_{M,mn}\varphi_{2} \right)}^2 +{\left(j_{M,mn}\varphi_{2}+ j_{M,mn}\varphi_3 \right) }^2\right\} \right),\label{E24}\\
&{\left[{\bf C}\right]}_{mn}
\!=\!{\left[{\bf \hat C}\right]}_{mn}
\exp \!\left(\!-\!\frac{1}{2}\sigma_{\text{est}}^2\left\{ {\left(i_{M,m} \varphi_1\!+\!j_{M,m} \varphi_{2} \!-\!i_{N,n} \right)}^2
\!+\!{\left(i_{M,m} \varphi_{2} \!+\!j_{M,m} \varphi_3\! -\!j_{N,n}\right)}^2\right\} \!\right)\label{E25},
\end{align}
 where $i_{M,mn} \triangleq \left(i_{M,n}-i_{M,m}\right)$, $j_{M,mn} \triangleq \left(j_{M,n}-j_{M,m}\right)$, $ {\bf \hat A} \triangleq   {\bf a}^*\left(\hat{\bar\theta}_{\text{x-B2U}},\hat{\bar\theta}_{\text{y-B2U}} \right)  {\bf a}^T \left(\hat{\bar\theta}_{\text{x-B2U}},\hat{\bar\theta}_{\text{y-B2U}} \right) $,
$ {\bf \hat B} \triangleq  {\bf b}^*\left(\hat{\bar\theta}_{\text{x-I2U}},\hat{\bar\theta}_{\text{y-I2U}} \right) {\bf b}^T\left(\hat{\bar\theta}_{\text{x-I2U}},\hat{\bar\theta}_{\text{y-I2U}} \right)  $ and
${\bf \hat C} \triangleq   {\bf b}^*\left(\hat{\bar\theta}_{\text{x-I2U}},\hat{\bar\theta}_{\text{y-I2U}} \right)    {\bf a}^T \left(\hat{\bar\theta}_{\text{x-B2U}},\hat{\bar\theta}_{\text{y-B2U}} \right)   $.
\end{proposition}
\begin{proof}
See Appendix \ref{A4}.
\end{proof}

Then the optimization problem can be formulated as:
\begin{align} {\label{E13}}
  \begin{array}{ll}
    \max\limits_{\left\{ \bm{\Theta}, {\bf w} \right\}}
& P_{r}={\bf w}^H{\bf T}{\bf w},  \\
 \operatorname{s.t.} & \begin{array}[t]{lll}
              {\left\|\bf w \right\|}^2 \le P_{BS},\\
             \left|{\left[ \bf \Theta \right]}_{ii}\right| =1,i=1,...,M,
           \end{array}
  \end{array}
\end{align}
where $P_\text{BS}$ is the maximum transmit power at the BS.

Since  the phase shift matrix ${\bm \Theta}$ and the BS beam ${\bf w}$ are coupled, the above optimization problem (\ref{E13}) is difficult to
solve. Responding to this, we will first optimize ${\bf w}$ by fixing
${\bm \Theta}$ and then update ${\bm \Theta}$ by fixing ${\bf w}$ respectively. As such, we can
obtain sub-optimal solutions for both ${\bf w}$ and ${\bm \Theta}$ by performing
this successive refinement process iteratively.

\subsection{BS beamforming}\label{sb1}
 By fixing the phase shift matrix ${\bm \Theta}$, we can simplify (\ref{E13})  to
\begin{align}
  \begin{array}{ll}
    \max\limits_{\left\{  {\bf w} \right\}}
& P\left( {\bf w} \right) ={\bf w}^H {\bf T} {\bf w},  \\
 \operatorname{s.t.} & \begin{array}[t]{lll}
                {\left\|\bf w \right\|}^2  \le P_{\text{BS}}
           \end{array}.
  \end{array}
\end{align}

The above problem is a convex problem  equivalent to
 \begin{align}
  \begin{array}{ll}
    \min\limits_{\left\{   {\bf w} \right\}}
&L\left( {\bf w},\mu \right)= -{\bf w}^H {\bf T} {\bf w}+\mu \left( {\left|\bf w \right|}^2-P_{\text{BS}}\right),  \\
 \operatorname{s.t.} & \begin{array}[t]{lll}
              \mu>0,
           \end{array}
  \end{array}
\end{align}
  where $\mu>0$ is the Lagrangian multiplier.

The   Karush-Kuhn-Tucker (KKT) conditions are
\begin{align}
&\nabla_{{\bf w}} L\left( {\bf w},\mu \right)=-2{\bf T w}+2\mu {\bf w}=0,\\
&\nabla_{{\mu}} L\left( {\bf w},\mu \right)={\left|\bf w \right|}^2-P_{\text{BS}}=0,
\end{align}
based on which, we have
\begin{align} \label{E16}
{\bf w}=\sqrt{P_{\text{BS}}} {\bf t}_{\text{max}},
\end{align}
where ${\bf t}_{\text{max}}$ is the eigenvector of ${\bf T}$ corresponding to the largest eigenvalue.

\subsection{IRS beamforming}\label{sb2}
For a given ${\bf w}$, the optimization problem (\ref{E13}) can be rewritten as
\begin{align}
  \begin{array}{ll}
    \max\limits_{\left\{ \bm{\Theta} \right\}}
 P{\left(\bm{\Theta}  \right)}& \triangleq
\beta_{\text{B2I2U}} {\bf w}^H
\sqrt{\beta_{\text{B2I2U}}  \beta_{\text{B2U}} } {\left( {\bm \Theta} \bar{\bf H}_{\text{B2I}}\right)}^H {\bf B}{\bm \Theta} \bar{\bf H}_{\text{B2I}} {\bf w} \\
&+\sqrt{\beta_{\text{B2I2U}}  \beta_{\text{B2U}} } \left( {\bf w}^H {\left( {\bm \Theta} \bar{\bf H}_{\text{B2I}}\right)}^H {\bf C}{\bf w}+ {\bf w}^H{\bf C}^H  {\bm \Theta} \bar{\bf H}_{\text{B2I}}{\bf w}\right) \\
&+ {\beta}_{\text{B2U}} {\bf w}^H{\bf A}{\bf w}+\sigma_{\text{NLOS}}^2 {\bf w}^H {\bf w}\\
 \operatorname{s.t.} & \begin{array}[t]{lll}
             \left|{\left[ \bf \Theta \right]}_{ii}\right| =1,i=1,...,M.
           \end{array}
  \end{array}
\end{align}

According to the rule that ${\bf E}^H {\bf X}^H {\bf F} = {\bf x}^H \left( {\bf E}^* \odot {\bf F} \right)$ with ${\bf X}={\text {diag}} \left({\bf x} \right)$, the objective function can be rewritten as
\begin{align}
P{\left(\bm{\Theta}  \right)}
&=\beta_{\text{B2I2U}}
{\text{tr}}\left(   {\bf \Theta}^H {\bf B} {\bf \Theta}  \bar{\bf H}_{\text{B2I}} {\bf w} {\bf w}^H
  \bar{\bf H}_{\text{B2I}}^H \right) \\
 &+ \sqrt{\beta_{\text{B2I2U}}  \beta_{\text{B2U}} } {\bm \xi}^H \left( \bar{\bf H}_{\text{B2I}}^*{\bf w}^* \odot {\bf C w} \right)
 + \sqrt{\beta_{\text{B2I2U}}  \beta_{\text{B2U}} }\left( \bar{\bf H}_{\text{B2I}}^T{\bf w}^T \odot {\bf  w}^H  {\bf C }^H \right)  {\bm \xi} \nonumber\\
 &+ {\beta}_{\text{B2U}} {\bf w}^H{\bf A}{\bf w}+\sigma_{\text{NLOS}}^2 {\bf w}^H {\bf w}.\nonumber
\end{align}

Using the property that  ${\bf X}^H {\bf E} {\bf Y} {\bf F} = {\bf x}^H \left( {\bf E} \odot {\bf F}^T \right) {\bf y}$ with ${\bf X}={\text {diag}} \left({\bf x} \right)$ and ${\bf Y}={\text {diag}} \left({\bf y} \right)$, the above equation can be further expressed as
\begin{align}
P{\left(\bm{\xi}  \right)}
&=\beta_{\text{B2I2U}} {\bm \xi}^H \left( {\bf B} \odot {\left(\bar{\bf H}_{\text{B2I}} {\bf w} {\bf w}^H
  \bar{\bf H}_{\text{B2I}}^H  \right)}^T \right) {\bm \xi}
    \\
 &+ \sqrt{\beta_{\text{B2I2U}}  \beta_{\text{B2U}} } {\bm \xi}^H \left( \bar{\bf H}_{\text{B2I}}^*{\bf w}^* \odot {\bf C w} \right)
 + \sqrt{\beta_{\text{B2I2U}}  \beta_{\text{B2U}} }\left( \bar{\bf H}_{\text{B2I}}^T{\bf w}^T \odot {\bf  w}^H  {\bf C }^H \right)  {\bm \xi}\nonumber\\
 &+ {\beta}_{\text{B2U}} {\bf w}^H{\bf A}{\bf w}+\sigma_{\text{NLOS}}^2 {\bf w}^H {\bf w},\nonumber
\end{align}
based on which, the optimization problem can be reformulated as
\begin{align} {\label{E14}}
  \begin{array}{ll}
    \max\limits_{\left\{ \bm{\xi} \right\}}
 P{\left(\bm{\xi}  \right)}
&=\beta_{\text{B2I2U}} {\bm \xi}^H \left( {\bf B} \odot {\left(\bar{\bf H}_{\text{B2I}} {\bf w} {\bf w}^H
  \bar{\bf H}_{\text{B2I}}^H  \right)}^T \right) {\bm \xi}
    \\
 &+ \sqrt{\beta_{\text{B2I2U}}  \beta_{\text{B2U}} } {\bm \xi}^H \left( \bar{\bf H}_{\text{B2I}}^*{\bf w}^* \odot {\bf C w} \right)
 + \sqrt{\beta_{\text{B2I2U}}  \beta_{\text{B2U}} }\left( \bar{\bf H}_{\text{B2I}}^T{\bf w}^T \odot {\bf  w}^H  {\bf C }^H \right)  {\bm \xi} \\
 &+ {\beta}_{\text{B2U}} {\bf w}^H{\bf A}{\bf w}+\sigma_{\text{NLOS}}^2 {\bf w}^H {\bf w} \\
 \operatorname{s.t.} & \begin{array}[t]{lll}
             \left|{ \xi }_{i}\right| =1,i=1,...,M.
           \end{array}
  \end{array}
\end{align}

Since $\left|{ \xi }_{i}\right| =1$, we have
${\text{tr}}\left({\bm \xi}{\bm \xi}^H \right)=M$. In order
to deal with the non-convex constraint of  $\left|{ \xi }_{i}\right| =1$, we relax
the problem (\ref{E14}) into the following optimization  problem  with a
convex $\ell_{\infty} $ constraint:
\begin{align} \label{E15}
  \begin{array}{lll}
    \max\limits_{\left\{ \bm{\xi} \right\}}
 &P{\left(\bm{\xi}  \right)}\\
 \operatorname{s.t.} & \begin{array}[t]{lll}
             {\text{tr}}\left({\bm \xi}{\bm \xi}^H \right)=M,\\
             {\left\| \bm \xi\right\|}_{\infty} \le 1.
           \end{array}
  \end{array}
\end{align}

Since the $\ell_{\infty}$ constraint is non-differentiable, we instead use
the $\ell_{p}$  with a  large $p$ to approximate $\ell_{\infty}$,
$\underset{p\to \infty}{\lim} \ell_{p}=\ell_{\infty}$. To solve
(\ref{E15}), we incorporate the second constraint
 by  exploiting the barrier method  with the logarithmic barrier function
$F\left(x \right)$ to approximate the penalty of violating the $\ell_{p}$ constraint,
$$ F(x)=\left\{
\begin{aligned}
-\frac{1}{\kappa} \ln \left( x\right), x>0 \\
\infty, x\le 0
\end{aligned}
\right.,
$$
where $\kappa$ is used to scale the barrier function penalty. As such,
the optimization problem can be rewritten as:
\begin{align} \label{op1}
  \begin{array}{ll}
    \min\limits_{\left\{ \bm{\xi} \right\}}
 G{\left(\bm{\xi}  \right)}
&=F\left( 1-{\left\| \bm \xi\right\|}_{p} \right)
 -\beta_{\text{B2I2U}} {\bm \xi}^H \left( {\bf B} \odot {\left(\bar{\bf H}_{\text{B2I}} {\bf w} {\bf w}^H
  \bar{\bf H}_{\text{B2I}}^H  \right)}^T \right) {\bm \xi}
    \\
 &- \sqrt{\beta_{\text{B2I2U}}  \beta_{\text{B2U}} } {\bm \xi}^H \left( \bar{\bf H}_{\text{B2I}}^*{\bf w}^* \odot {\bf C w} \right)
  -\sqrt{\beta_{\text{B2I2U}}  \beta_{\text{B2U}} }\left( \bar{\bf H}_{\text{B2I}}^T{\bf w}^T \odot {\bf  w}^H  {\bf C }^H \right)  {\bm \xi}  \\
 \operatorname{s.t.} & \begin{array}[t]{lll}
             {\text{tr}}\left({\bm \xi}{\bm \xi}^H \right)=M.
           \end{array}
  \end{array}
\end{align}

  Due to the non-convex constraint $\text{tr}({\bm \xi} {\bm \xi}^H )=M$, the optimization problem (\ref{op1}) is non-convex. Next, we try to solve the above optimization problem by exploiting a gradient method. As such, sub-optimal solutions can be obtained.
The gradient of the objective function $ G{\left(\bm{\xi}  \right)}$ can be calculated as
\begin{align}
\nabla_{{\bm \xi}} G{\left(\bm{\xi}  \right)}= \frac{ {\left\| \bm \xi\right\|}_{p}^{1-p} {\bm \zeta} }{2 \kappa \left( 1-{\left\| \bm \xi\right\|}_{p} \right)}-\nabla_{{\bm \xi}} P{\left(\bm{\xi}  \right)},
\end{align}
where
\begin{align}\label{d5}
{\bm \zeta}={\left[\xi_1 {\left| \xi_1 \right|}^{p-2}, \xi_2 {\left| \xi_2 \right|}^{p-2}, ...,\xi_M {\left| \xi_M \right|}^{p-2}  \right]}^T.
\end{align}

And the gradient $\nabla_{{\bm \xi}} P{\left(\bm{\xi}  \right)}$ can be computed as
\begin{align}
\nabla_{{\bm \xi}} P{\left(\bm{\xi}  \right)}\!=\!
2 \beta_{\text{B2I2U}} \left( {\bf B} \odot {\left(\bar{\bf H}_{\text{B2I}} {\bf w} {\bf w}^H
  \bar{\bf H}_{\text{B2I}}^H  \right)}^T\right){\bm \xi}
 \! +\!2 \sqrt{\beta_{\text{B2I2U}} \beta_{\text{B2U}}} {\text{Re}} \left(  \left( \bar{\bf H}_{\text{B2I}}^*{\bf w}^* \odot {\bf C w} \right) {\bf 1}_{N} \right),
\end{align}
where ${\bf 1}_{N}$ is a $N \times 1 $ vector whose elements all equal to 1.

Thus, $\nabla_{{\bm \xi}} G{\left(\bm{\xi}  \right)}$ can be rewritten as
\begin{align} \label{E18}
\nabla_{{\bm \xi}} G{\left(\bm{\xi}  \right)}& =
-2 \beta_{\text{B2I2U}} \left( {\bf B} \odot {\left(\bar{\bf H}_{\text{B2I}} {\bf w} {\bf w}^H
  \bar{\bf H}_{\text{B2I}}^H  \right)}^T\right) {\bm \xi} \\
&   -2 \sqrt{\beta_{\text{B2I2U}} \beta_{\text{B2U}}} {\text{Re}} \left(  \left( \bar{\bf H}_{\text{B2I}}^*{\bf w}^* \odot {\bf C w} \right) {\bf 1}_{N} \right)
+\frac{ {\left\| \bm \xi\right\|}_{p}^{1-p} {\bm \zeta} }{2 \kappa \left( 1-{\left\| \bm \xi\right\|}_{p} \right)}. \nonumber
\end{align}

 Due to the constraint $\text{tr}\left({\bm \xi}{\bm \xi}^H \right)=M$, we project ${\bf g}_{\text{gd}}=-\nabla_{{\bm \xi}} G{\left(\bm{\xi}  \right)}$ into the tangent plane of $\text{tr}\left({\bm \xi}{\bm \xi}^H \right)=M$:
  \begin{align} \label{E19}
  {\bf g}_{\text{p}}={\bf g}_{\text{gd}}-\frac{{\bf g}_{\text{gd}}^T{\bm \xi}^* {\bm \xi}}{{\left\| {\bm \xi} \right\|}^2}.
  \end{align}

  Then we use ${\bf g}_{\text{p}}$ as the search direction.
  Algorithm \ref{al1} provides the pseudo-code for the process. For convenience, we  collect the principal and important parameters and variables in Table I.

  \begin{proposition} \label{p5}
 The constant-modulus constraint $|\xi_i|=1,i=1,\dots,M$ is satisfied.
  \end{proposition}
  \begin{proof}
  See Appendix \ref{A7}.
  \end{proof}

 \begin{algorithm}[!t]         
\caption{IRS Beamforming Algorithm}            
\label{al1}
\begin{algorithmic}[1]
\State $\mathbf{Initialization:}$ Given a feasible phase shift vector ${\bm \xi}_{1}$, a large $p>0$,
the iteration index $i=0$, the maximum iterations $N_{\text{iter}}>0$, halting criterion $\varepsilon>0$ and the barrier coefficient $\kappa>0$.

\Repeat
\State $i \leftarrow i+1$.
  \State  Compute the gradient as per (\ref{E18}).
  \State  Compute the search direction $ {\bf g}_{\text{p}}$ as per (\ref{E19}).

    \State  For $0 \le \varpi \le 1$, searching for it by
   \begin{align}
   \varpi ^{\star}= \underset{\varpi}{\arg \max} \ P\left( \left(1-\varpi \right) {\bm \xi}_{i}+\varpi \sqrt{M} \frac{{\bf g}_{\text{p}}}{{\left\| {\bf g}_{\text{p}}\right\|}^2} \right).
     \end{align}

     \State Update:
     \begin{align} \label{E17}
     {\bm \xi}_{i+1}=\left(1-\varpi^{\star} \right) {\bm \xi}_{i}+\varpi^{\star} \sqrt{M} \frac{{\bf g}_{\text{p}}}{{\left\| {\bf g}_{\text{p}}\right\|}^2}.
     \end{align}

\Until ${\left| P\left({\bm \xi}_{i+1} \right)- P\left({\bm \xi}_{i} \right) \right|} \le \varepsilon$ or $i \ge N_{\text{iter}}$.

\State $\mathbf{Output:}$ ${\bm \xi}^{\star} =\exp\left[j \text{angle} ({\bm \xi_{i}}) \right]$.
  \end{algorithmic}
\end{algorithm}

\begin{table}
\centering
\label{table1}
\begin{tabular} {|c|c|}
 \hline
 {\bf Parameter} & {\bf Definition} \\
 \hline
Variable ${\bf w} $ & Transmit beamforming vector.  \\
 \hline
Variable ${\bm \xi}$  & Phase shift vector. \\
 \hline
Variable ${\bm \Theta}$  & Phase shift matrix ${\bm \Theta}=\text{diag}({\bm \xi})$  \\
 \hline
 ${\beta}_{\text{B2I2U}}$,${\beta}_{\text{B2U}}$  & Defined by (\ref{d2}) and (\ref{d3}),respectively. \\
 \hline
${\bf B}$,${\bf C}$ & Defined by (\ref{E24}) and (\ref{E25}), respectively. \\
 \hline
$\bar{\bf H}_\text{B2I}$   &  Defined by (\ref{d1}).\\
 \hline
 \end{tabular}
\caption{   Parameter Table.}
\end{table}

\subsection{Joint optimization of BS beamforming and IRS beamforming }
According to \ref{sb1} and \ref{sb2}, a joint optimization scheme is developed in Algorithm \ref{al2}, where the BS beam and phase shift beam are iteratively optimized.

  \begin{algorithm}[!t]         
\caption{Joint Optimization Algorithm}            
\label{al2}
\begin{algorithmic}[1]
\State $\mathbf{Initialization:}$ Given feasible initial solutions ${\bm \xi}_{0}$, ${\bf w}_{0}$ and the iteration index $i=0$.

\Repeat

\State Perform BS beamforming: According to ${\bm \xi}_{0}$, optimize the BS beam by invoking the result (\ref{E16}) in \ref{sb1},  which yields  ${\bf w}_{i+1}$.
\State Perform IRS beamforming: Based on ${\bf w}_{i+1}$, optimize the transmit beamformer via Algorithm \ref{al1}, which yields ${\bm \xi}_{i+1}$.

\State $i \leftarrow i+1$.

\Until ${\bm \xi}_{i}$ and ${\bf w}_{i}$ are converged.

\State $\mathbf{Output:}$ ${\bm \xi}^{\star} ={\bm \xi}_i$ and ${\bf w}^{\star}={\bf w}_{i}$.

  \end{algorithmic}
\end{algorithm}

\begin{proposition} \label{p4}
The joint optimization algorithm always guarantees $P\left( {\bf w}_{i+1},{\bm \xi}_{i+1}\right) \ge P\left( {\bf w}_{i},{\bm \xi}_{i}\right) $.
\end{proposition}

\begin{proof}
1) Proof of $P\left( {\bf w}_{i+1},{\bm \xi}_{i+1}\right) \ge P\left( {\bf w}_{i},{\bm \xi}_{i+1}\right)$

By fixing ${\bm \xi}$, the optimization of ${\bf w}$ is a convex problem. And, ${\bf w}_{i+1}$ is the optimal solution corresponding to the phase shift beam ${\bm \xi}_{i+1}$.
Thus, we have
\begin{align}
P\left( {\bf w}_{i+1},{\bm \xi}_{i+1}\right) \ge P\left( {\bf w}_{i},{\bm \xi}_{i+1}\right).
\end{align}

2) Proof of $P\left( {\bf w}_{i+1},{\bm \xi}_{i+1}\right) \ge P\left( {\bf w}_{i+1},{\bm \xi}_{i}\right)$

Fixing ${\bf w}_{i+1}$,  the Taylor expansion of $P\left( {\bf w}_{i+1},{\bm \xi}_{i+1}\right)$ can be expressed as
\begin{align}
P\left( {\bf w}_{i+1},{\bm \xi}_{i+1}\right)&=
P\left({\bf w}_{i+1}, {\bm \xi}_{i}\right)+ {\left( \nabla_{{\bm \xi}_{i}} P\right)}^H
\left\{ P\left( {\bf w}_{i+1},{\bm \xi}_{i+1}\right)-P\left( {\bf w}_{i+1},{\bm \xi}_{i}\right)  \right\} \\
&+ o\left\{ {\left( P\left( {\bf w}_{i+1},{\bm \xi}_{i+1}\right)-P\left( {\bf w}_{i+1},{\bm \xi}_{i}\right) \right)}^2 \right\}. \nonumber
\end{align}

Based on (\ref{E17}) in Algorithm \ref{al1}, for any $\varpi \to 0$, the above equation can be written as
\begin{align} \label{E22}
P\left( {\bf w}_{i+1},{\bm \xi}_{i+1}\right)&=
P\left({\bf w}_{i+1}, {\bm \xi}_{i}\right)+ {\left( \nabla_{{\bm \xi}_{i}} P\right)}^H
\sqrt{M} \frac{{\bf g}_{\text p} }{{\left\|{\bf g}_{\text p} \right\|}^2} \varpi^{\star}
+ o\left\{ {\varpi^{\star}}^2 \right\}\\
& \approx  P\left({\bf w}_{i+1}, {\bm \xi}_{i}\right)+ {\left( \nabla_{{\bm \xi}_{i}} P\right)}^H
{\bf g}_{\text p}  \frac{\sqrt{M} }{{\left\|{\bf g}_{\text p} \right\|}^2} \varpi^{\star}. \nonumber
\end{align}

Next, we focus on the calculation of ${\left( \nabla_{{\bm \xi}_{i}} P\right)}^H
{\bf g}_{\text p}$.
\begin{align}
{\left( \nabla_{{\bm \xi}_{i}} P\right)}^H
{\bf g}_{\text p}
={\left({\bf g}_{\text{gd}} + \nabla_{{\bm \xi}_{i}}F \left(1- {\left\|{\bm \xi}_i \right\|}^2 \right) \right)}^H {\bf g}_{\text p}.
\end{align}

  Recall that
$
\nabla_{{\bm \xi}_{i}}F \left(1- {\left\|{\bm \xi}_i \right\|}^2 \right) =\frac{ {\left\| \bm \xi_i\right\|}_{p}^{1-p} {\bm \Omega \bm \xi_{i}} }{2 \kappa \left( 1-{\left\| \bm \xi_i\right\|}_{p} \right)}
$
with
$
{\bm \Omega}= \text{diag} \left( {\left| \xi_1 \right|}^{p-2},   {\left| \xi_2 \right|}^{p-2}, ..., {\left| \xi_M \right|}^{p-2} \right)
$, and ${\bf g}_{\text{p}}={\bf g}_{\text{gd}}-\frac{{\bf g}_{\text{gd}}^T{\bm \xi_i}^* {\bm \xi_i}}{{\left\| {\bm \xi_i} \right\|}^2}$. We have ${\left( \nabla_{{\bm \xi}_{i}}F \left(1- {\left\|{\bm \xi}_i \right\|}^2 \right) \right)}^H {\bf g}_{\text p}=0$, and thus the above equation can be calculated as
\begin{align} \label{E23}
{\left( \nabla_{{\bm \xi}_{i}} P\right)}^H
{\bf g}_{\text p}
={\bf g}_{\text{gd}}^H {\bf g}_{\text p}
={\left\| {\bf g}_{\text{gd}} \right\|}^2 \left(1-{\left(\cos\varrho \right)}^2 \right) \ge 0,
\end{align}
where $\varrho= \text{arcos} \left( \frac{ {\bf g}_{\text{gd}}^H {\bm \xi}_i}
{\left\| {\bf g}_{\text{gd}} \right\| \left\| {\bm \xi}_i \right\|} \right)$ is the angle between vectors ${\bf g}_{\text{gd}} $ and ${\bm \xi}_i$.

Substituting (\ref{E23}) into (\ref{E22}),  we obtain
\begin{align}
P\left( {\bf w}_{i+1},{\bm \xi}_{i+1}\right) -P\left({\bf w}_{i+1}, {\bm \xi}_{i}\right)
={\left( \nabla_{{\bm \xi}_{i}} P\right)}^H
{\bf g}_{\text p}  \frac{\sqrt{M} }{{\left\|{\bf g}_{\text p} \right\|}^2} \varpi^{\star} \ge 0.
\end{align}

To this end, combining 1) and 2), we have
\begin{align}
P\left( {\bf w}_{i+1},{\bm \xi}_{i+1}\right) \ge P\left( {\bf w}_{i},{\bm \xi}_{i+1}\right) \ge
P\left( {\bf w}_{i},{\bm \xi}_{i}\right) .
\end{align}

\end{proof}

\begin{remark}
Our method of iteratively optimizing the transmit beam and the phase shift beam can provide an efficient way to  gradually increase the average received signal power, although  only  a
sub-optimal solution can be ensured due to the non-convexity of the problem.
\end{remark}

\begin{remark}

For the BS beamforming scheme, the complexity is consumed by the eigenvector calculation. Therefore, the complexity order of BS beamforming is $\mathcal{O}\left( N^3 \right)$.

For the IRS beamforming algorithm , the computational complexity is mainly determined by the calculation of the gradient (\ref{E18}), involving computing the $\ell_p$ norm and the matrix multiplication $\left( {\bf B} \odot {\left(\bar{\bf H}_{\text{B2I}} {\bf w} {\bf w}^H
  \bar{\bf H}_{\text{B2I}}^H  \right)}^T\right) {\bm \xi} $. Hence, the complexity order of Algorithm \ref{al1} for each iteration is $\mathcal{O} \left( M^2+pM\right)$.

Finally, the overall complexity order of the joint optimization algorithm  for each iteration is given by
$
\mathcal{O}\left( N_{\text{iter}} \left( M^2+pM\right) +N^3 \right),
$
where $N_{\text{iter}}$ is specified in Algorithm \ref{al1}.

\end{remark}

\section{Achievable Rate analysis}\label{s4}

In this section, we will present detailed analysis of the achievable rate.
\begin{theorem} \label{t3}
The achievable rate is given by
\begin{align} \label{E12}
R=\log_2\left(1+\frac{\tilde{\bf w}^H {\bf T} \tilde{\bf w}}{\sigma_{0}^2} \right),
\end{align}
where $\tilde{\bf w}  \triangleq \sqrt{P_\text{BS} } {\bf t}_{\text{max}}$.
\end{theorem}
\begin{proof}
Starting from Proposition \ref{p2}, the result is readily obtained.
\end{proof}

Theorem \ref{t3} presents an expression for the achievable rate which quantifies the impact
of key system parameters, such as the number of BS antennas and reflecting elements, as well as the impact of estimation error on the achievable rate. For instance, it can be seen that the SNR is related to $\bf T$.
 From the expression of $\bf T$ which is defined  in  Proposition 2, we can conclude that
the SNR drops nearly exponentially with the variance of the angle estimation error. This is because inaccurate estimated angles  makes it difficult to generate highly directional beams, thereby causing severe power loss.

\begin{corollary} \label{c2}
Assuming a large Rician K-factor and that the IRS and the user are in the similar direction, the achievable rate can be approximated as
\begin{align}
R_{\text{approx}}=\log_2 \left(1+\frac{ P_{\text{BS}} \Omega  }{\sigma_0^2} \right),
\end{align}
where
\begin{align}\label{E27}
\Omega & \triangleq
N \beta_{\text{B2I2U}} \sum\limits_{m=1}^{M} \sum\limits_{n=1}^{M} { \xi}_{m}^{*}{ \xi}_{n} {\left[ \bf B \right]}_{mn} b_{\text{B2Ia},m}^* b_{\text{B2Ia},n}\\
&+2\sqrt{\beta_{\text{B2I2U}}  \beta_{\text{B2U}} } \text{Re} \left\{ \sum\limits_{m=1}^{M} { \xi}_{m}^{*}
  b_{\text{B2Ia},m}^* \left(\sum\limits_{i=1}^{N} {\left[ \bf C \right]}_{mi}  a_{\text{B2I},i}^*  \right)\right\}+N \beta_{\text{B2U}}+N \sigma_{\text{NLOS}}^2, \nonumber
\end{align}
where $a_{\text{B2I},m}$ and $b_{\text{B2Ia},m}$ are the $m$-th elements of ${\bf a}^T \left(\bar \theta_{\text{x-B2I}},\bar\theta_{\text{y-B2I}} \right)$ and ${\bf b}^T \left(\bar \theta_{\text{x-B2Ia}},\bar\theta_{\text{y-B2Ia}} \right)$, respectively.
\end{corollary}
\begin{proof}
See Appendix \ref{A5}.
\end{proof}

Corollary \ref{c2} implies that increasing the number of BS antennas can greatly improve the achievable rate. In particular, both the BS-user link and BS-IRS-user link SNRs grow linearly with $N$.
To get more insights, we derive an upper bound for the achievable rate.

\begin{proposition}\label{p3}
An upper bound for the achievable rate is given by
\begin{align} \label{E28}
R_{\text{upper}}=\log_2 \left(1+\frac{ P_{\text{BS}} N \left(\beta_{\text{B2I2U}}M^2+2 \sqrt{\beta_{\text{B2I2U}}  \beta_{\text{B2U}} }M+  \beta_{\text{B2U}}+ \sigma^2_{\text{NLOS}} \right)  }{\sigma_0^2} \right).
\end{align}
\end{proposition}

\begin{proof}
Starting from Corollary \ref{c2} and using the fact that $\sum\limits_{m=1}^{M} a_m \le \sum\limits_{m=1}^{M} \left|a_m \right| $, we have
\begin{align}
\Omega \le
N \beta_{\text{B2I2U}} \sum\limits_{m=1}^{M} \sum\limits_{n= 1}^{M} \left| {\left[ \bf B \right]}_{mn} \right|
+2\sqrt{\beta_{\text{B2I2U}}  \beta_{\text{B2U}} }   \sum\limits_{m=1}^{M}  \sum\limits_{i=1}^{N} \left|{\left[ \bf C \right]}_{mi}\right|     +N \beta_{\text{B2U}}+N \sigma_{\text{NLOS}}^2.
\end{align}

Noticing that $\left| {\left[ \bf B \right]}_{mn} \right| \le 1$ and $\left|{\left[ \bf C \right]}_{mi}\right| \le 1$, (\ref{E28}) can be proved.
\end{proof}

Proposition \ref{p3} indicates that with fixed transmit power, the achievable rate is mainly determined by the numbers of reflecting elements and BS antennas. Specifically, the effective SNR is proportional to the number of BS antennas.  Moreover, there is a gain $M^2$ corresponding to the BS-IRS-user link. This is reasonable because the IRS not only achieves the phase shift beamforming gain of order $M$ in the IRS-user link, but also captures an inherent aperture gain of order $M$ by collecting more signal power in the BS-IRS link. It is worth noting that this $M^2$ gain only holds when $\sqrt{M A_p } \ll d_\text{B2I}$ with $A_p$  being the effective aperture/area of each reflecting element, due to the law of energy conservation \cite{bjornson2019demystifying}.

\begin{remark}
This paper mainly focuses on a single user scenario. However, the proposed angle domain design framework can be easily extended to the multi-user case. For instance, with orthogonal pilot sequences, the proposed angle-domain estimation method can be directly applied to separately estimate the effective angles of each user. Also, the alternating optimization method can be similarly applied to reduce the complexity of beamforming algorithms.
The key difference is that, in a multi-user scenario, the co-channel interference should be taken into account. In addition, since all users share the same BS and IRS beamforming vectors, the resultant optimization problem becomes much more complicated.
\end{remark}

\section{numerical  results}\label{s5}
In this section, we provide numerical results to illustrate the performance of the angle-domain IRS-aided system, as well as to verify the performance of the proposed ML estimator and the joint optimization scheme.
  The considered system operates at $2.45$ GHz. The large-scale fading coefficient is modeled as $\alpha= L_{\text{d}}^{-\chi}$, where $\chi$ is the path loss exponent, and $L_{\text{d}}$ is the transmission distance.
We assume the BS is located at the origin. The  locations of the user and the IRS  are denoted by $(d_{\text{B2U}},  \theta_\text{B2U},\phi_\text{B2U})$ and  $(d_{\text{B2I}},  \theta_\text{B2I},\phi_\text{B2I})$, respectively.
For all
simulations, unless otherwise specified, the following setup is
used: $N=16, M=256,  P_\text{BS}=10 \ \text{dBm}, \sigma_0^2=-60 \  \text{dBm}, v_{\text{B2U}} =v_{\text{I2U}}=v_{\text{B2I}}=5, \chi_\text{B2U}=\chi_\text{I2U}=\chi_\text{B2I}=2.5$, IRS location $(42 \text{m},63^\circ ,-16^\circ )$ and user location $(41 \text{m},47^\circ ,-16^\circ )$.

  Fig. \ref{f1} illustrates the performance of the proposed ML estimator given by Theorem \ref{t1}, where the  analytical results are generated by (\ref{delta_est}) in Corollary \ref{c1}. For comparison, the conventional angle-domain  estimation  method, i.e., the MUSIC method \cite{Trees2003Detection},  is presented as the benchmark.
Note that $\text{MSE}_\text{x-B2U}$ and $\text{MSE}_\text{y-B2U}$ are defined by $\mathbb{E}\{{( \hat{\bar\theta}_\text{x-B2U}-\bar\theta_\text{x-B2U})}^2\}$ and $\mathbb{E}\{{( \hat{\bar\theta}_\text{y-B2U}-\bar\theta_\text{y-B2U})}^2\}$, respectively.
As can be readily observed, the numerical results match exactly with the analytical results, thereby validating the correctness of the analytical expressions.
Although the MUSIC method is slightly better than the proposed method, it requires more training time and has higher computational complexity  than the proposed method.
Moreover, as expected in Corollary \ref{c1}, increasing the number of BS antennas can greatly reduce the mean square error (MSE). This is because with a large number of BS antennas,  more angle-related data can be obtained, based on which we can estimate the angles more accurately. In addition, we can see that the MSE is a decreasing function with respect to the Rician K-factor, because a larger Rician K-factor means less uncertainty arising from NLOS paths. Also, as the received SNR becomes larger, the MSE gradually decreases due to less uncertainty caused by noise.
\begin{figure}[!ht]
  \centering
  \includegraphics[width=3.5 in]{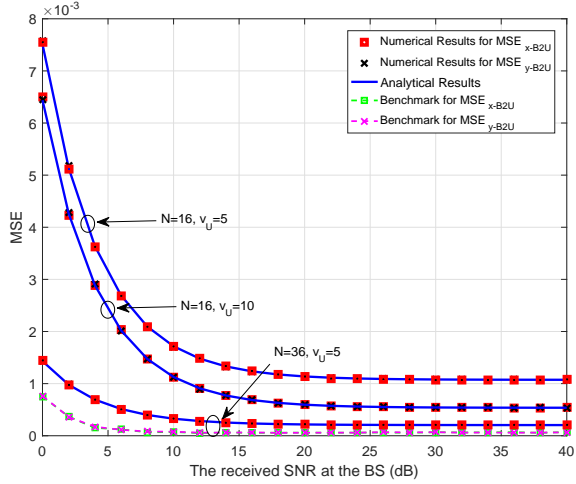}
   \caption{ Performance of the proposed ML
   estimator.}
  \label{f1}
\end{figure}

Fig. \ref{f2} illustrates the MSE for the calculated  effective angles from the IRS to the user given by Lemma \ref{L1}, where the  analytical results are generated by Theorem \ref{t2}.  Note that  $\text{MSE}_\text {x-I2U}$ and $\text{MSE}_\text{y-I2U}$ are defined by $\mathbb{E}\{{( \hat{\bar\theta}_\text{x-I2U}-\bar\theta_\text{x-I2U})}^2\}$ and $\mathbb{E}\{{( \hat{\bar\theta}_\text{y-I2U}-\bar\theta_\text{y-I2U})}^2\}$, respectively. We can see that the numerical curves match the analytical curves well.
Besides, the MSE decreases with the received SNR at the BS. This is intuitive because the calculation of  effective angles from IRS to the user relies on the estimated angles from the BS to the user.
Moreover, increasing the ratio $Ra$ would severely  degrade the accuracy of estimated  effective angles. The reason is that a larger $Ra$ means that the  user is  closer to the IRS, making the calculated angles more sensitive to the estimation error of angles  from the BS to the user .
\begin{figure}[!ht]
  \centering
  \includegraphics[width=3.5in]{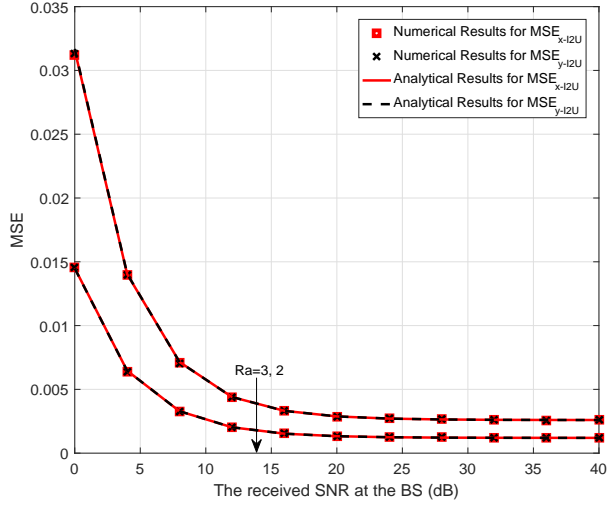}
   \caption{ The MSE for the calculated  effective angles from the IRS to the user with IRS location $(57.8 \text{m},63^\circ ,-16^\circ )$   .}
  \label{f2}
\end{figure}

Fig. \ref{f23}  illustrates the convergence of the proposed joint beamforming method given by Algorithm \ref{al2}. As expected in Proposition \ref{p4}, the proposed algorithm can gradually increase the average received SNR. Moreover, the algorithm converges after only several iterations, making it a low-complexity method. Besides, the received SNR increases as the number of reflecting elements  becomes larger, which indicates the benefit of applying a large number of reflecting elements.

\begin{figure}[!ht]
  \centering
  \includegraphics[width=3.5in]{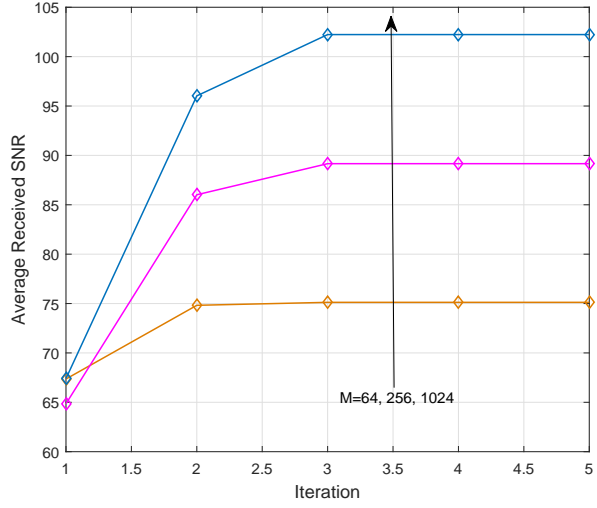}
  \caption{  The convergence of the proposed joint optimization algorithm.}
  \label{f23}
\end{figure}

Fig. \ref{f_beamforming} shows the performance of the proposed angle-based beamforming algorithm. For comparison, the algorithm in \cite{wu2019intelligent} which assumes full CSI is presented as ``Benchmark 1'', while a beamforming algorithm based on angles estimated by the   MUSIC method is presented as ``Benchmark 2''.
 As can be seen, our proposed angle-based algorithm  achieves nearly the same performance as both benchmark algorithms.
The beamforming algorithm corresponding to  ``Benchmark 1 '' achieves the best performance, but requires full CSI. Although the angle-based beamforming algorithm corresponding to  ``Benchmark 2 '' is slightly better than the proposed algorithm,  it adopts the MUSIC method to estimate angle information and thus has higher training overhead and computational complexity than the proposed angle-domain estimation method.

\begin{figure}[!ht]
  \centering
  \includegraphics[width=3.5 in]{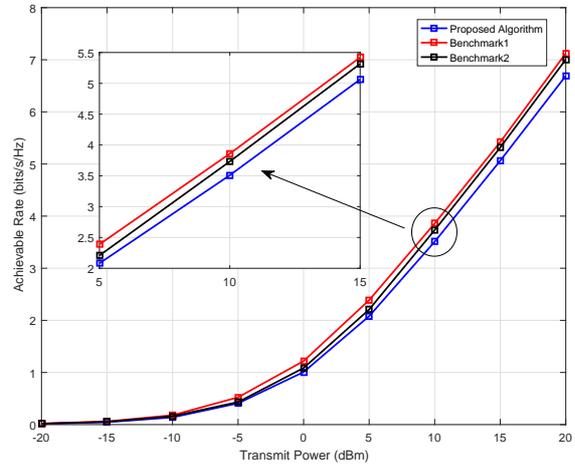}
   \caption{Performance of the proposed beamforming scheme with $N=4$ and $M=36$.}
  \label{f_beamforming}
\end{figure}

Fig. \ref{f4} depicts the impact of the number of reflecting elements on the beam pattern of the optimized BS beam. As can be observed, with $400$ reflecting elements, the main lobe is in the user direction. As the number of reflecting elements increases, the lobe  in the user direction gradually becomes smaller and the main lobe appears in the IRS direction, because increasing the number of reflecting elements can enhance the BS-IRS-user link.

\begin{figure}[htbp]
\centering
\begin{minipage}[t]{0.32\linewidth}
\centering
\subfigure { \label{f4_M2500}
\includegraphics[width=2.5in]{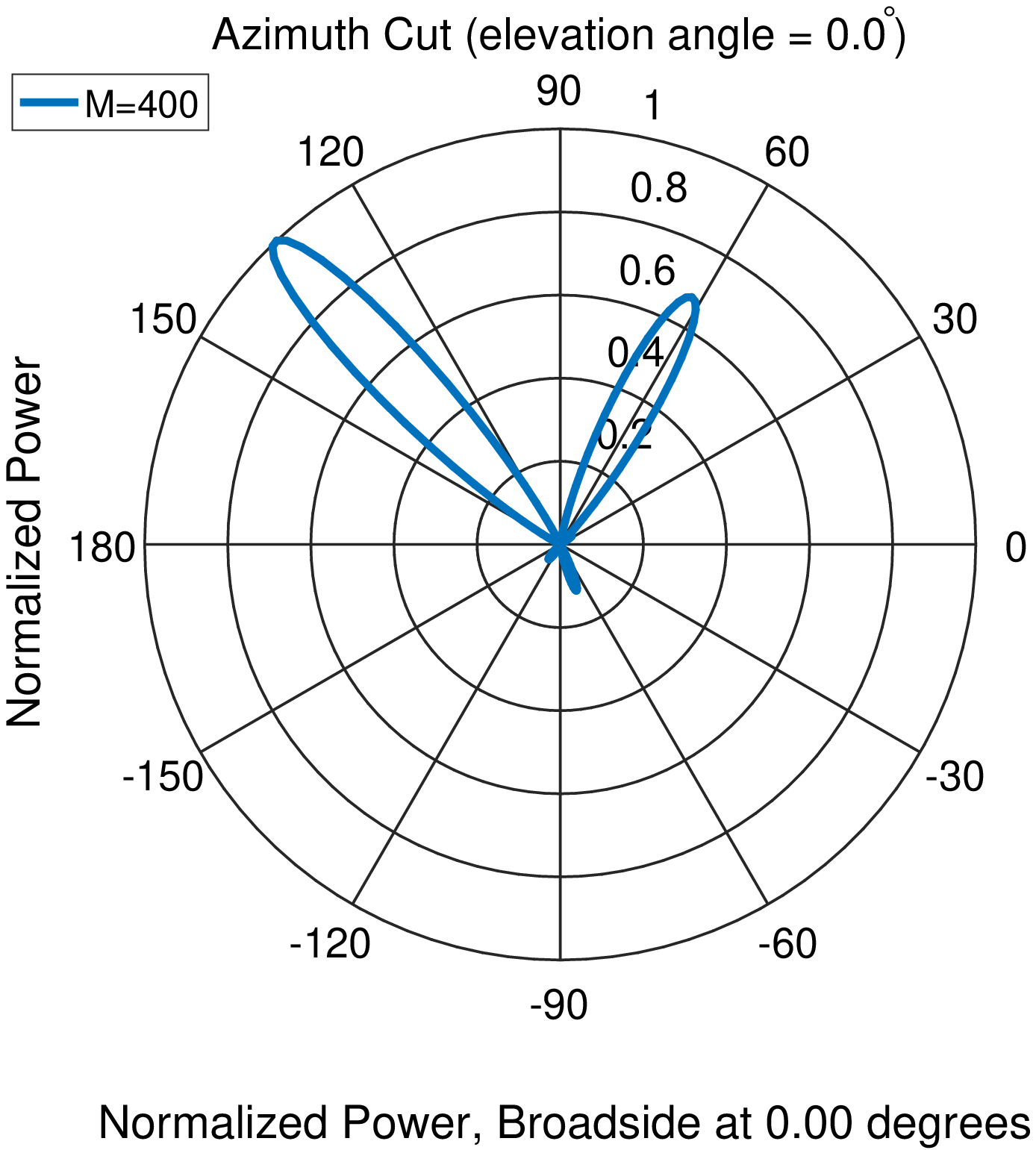}}
\end{minipage}
\begin{minipage}[t]{0.32\linewidth}
\centering
\subfigure{ \label{f4_M3600}
\includegraphics[width=2.5in]{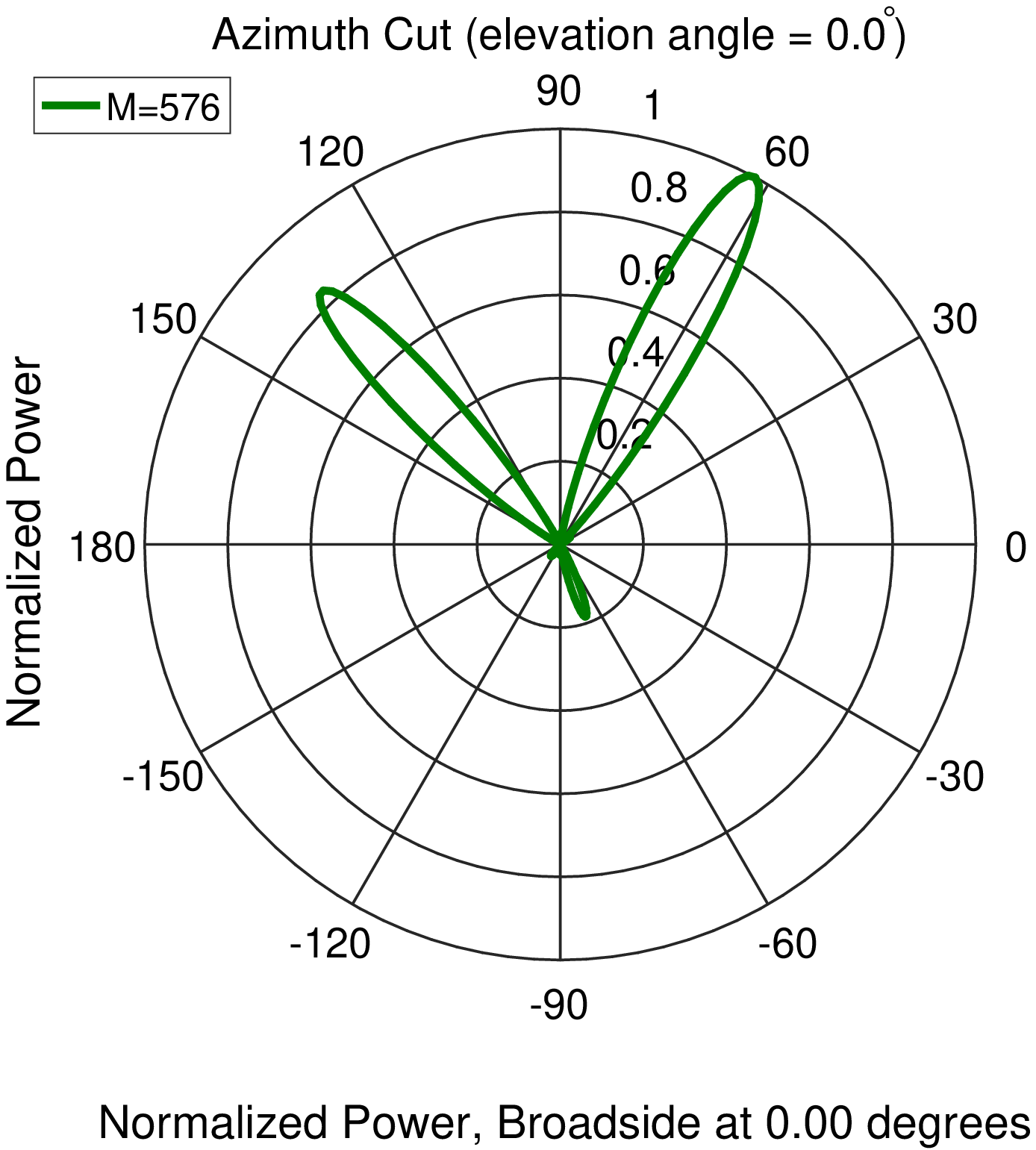}}
\end{minipage}
\centering
\begin{minipage}[t]{0.32\linewidth}
\centering
\subfigure{ \label{f4_M3600}
\includegraphics[width=2.5in]{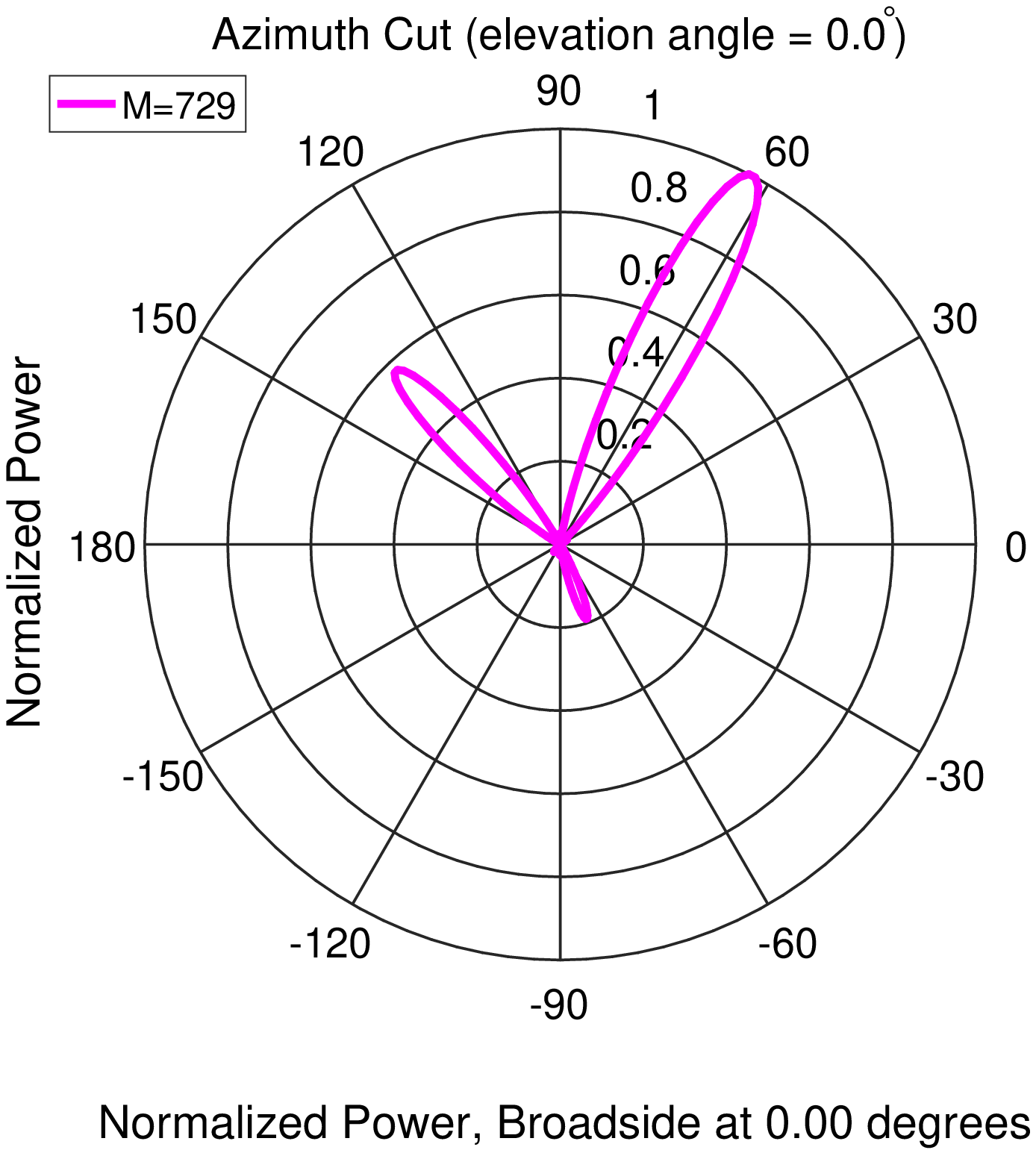}}
\end{minipage}
\caption{ The impact of the number of reflecting elements on the transmit beam pattern with $N=36$, user location $(41\text{m},133^\circ, -16^\circ)$ and IRS location $(42\text{m}, 63^\circ,-16^\circ)$.}
\label{f4}
\end{figure}

  Fig. \ref{f7} shows the achievable rate of the considered system with different configurations, where the curves associated with ``Approximate rate'' and ``Limit''  are plotted according to Corollary \ref{c2} and Proposition \ref{p3}, respectively. As can be readily observed, the three curves corresponding to  ``Limit'',``Approximate rate'' and ``With IRS and direct link'' have the similar trend, which verifies the effectiveness of our analysis
in  Corollary \ref{c2} and Proposition \ref{p3}.
Moreover, we can see that the achievable rate with  IRS is much larger than that without IRS, which indicates the great benefit of IRS.  Besides, when there is no direct link, the achievable rate becomes extremely low. This is because the acquisition of  angle information corresponding to both the BS-IRS and IRS-user links relies on the BS-user direct link. Without the direct link, the BS is not able to obtain any channel information.

\begin{figure}[!ht]
  \centering
  \includegraphics[width=3.5in]{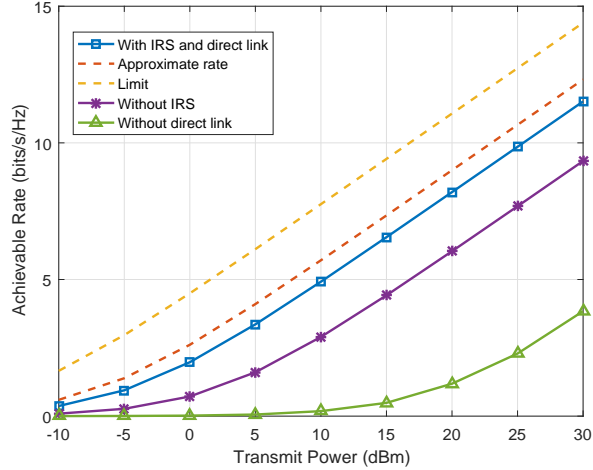}
   \caption{  Achievable rate performance with  $N=4, \chi_\text{B2I}=\chi_\text{I2U}=2.3$ and $\chi_\text{B2U}=2.8$.}
  \label{f7}
\end{figure}

\section{conclusion}\label{s6}
This paper considers an IRS-aided system from an angle-domain aspect. The ML estimators for the  effective angles from the BS to the user are provided, based on which the  effective angles from the IRS to the user are calculated.
It has been shown that increasing the number of BS antennas can significantly reduce the estimation error. Also,  placing the IRS closer to the BS would lead to a smaller estimation error of  the estimated angles from the IRS to the user. Then, exploiting the estimated angles, a joint optimization algorithm of BS beamforming and IRS beamforming has been proposed, which achieves  similar performance to two   benchmark algorithms based on full CSI  and the MUSIC method respectively. Beam patterns
of the optimized BS beam indicate that as the number of reflecting elements becomes larger, the beam becomes more focused towards the IRS direction, as should be expected.
Analysis of the achievable rate quantifies the benefit of deploying a large number of BS antennas or reflecting elements. In particular,
the BS-user link and the BS-IRS-user link can obtain power gains of order $N$ and $NM^2$, respectively.

\begin{appendices}
\section{Proof of Proposition \ref{p1}} \label{A1}

The received signal at the $n$-th antenna can be decomposed into
the LOS component and the  uncertainty component, i.e.,
\begin{align}
r_n=A_{\text{ob},n} e^{j \vartheta_n}=
A_{\text{LOS},n} e^{j \left(  \theta_{\text{q}}+ i_{N,n} \bar\theta_{\text {x-B2I}}+j_{N,n}  \bar\theta_{\text {y-B2I}}  \right)} + A_{\text{unct},n} e^{j \vartheta_{\text{unct},n}},
\end{align}
where
\begin{align}
 & A_{\text{LOS},n} e^{j \left(  \theta_{\text{q}}+ i_{N,n} \bar\theta_{\text {x-B2I}}+j_{N,n}  \bar\theta_{\text {y-B2I}}  \right)} \triangleq \sqrt{\frac{\alpha_{\text{U}} v_{\text{B2U}} }{v_{\text{B2U}}+1}} e^{\left(
 i_{N,n}  \bar\theta_{\text {x-B2I}}+j_{N,n}  \bar\theta_{\text {y-B2I}}\right)} q,\\
 &A_{\text{unct},n} e^{j \vartheta_{\text{unct},n}} \triangleq \sqrt{\frac{\alpha_{\text{U}} }{v_{\text{B2U}}+1}}\tilde{h}_{\text{B2U},n} q +n_{\text{BS},n},
\end{align}
 with $A_{\text{ob},n}, A_{\text{LOS},n}$ and $A_{\text{unct},n}$ denoting the corresponding amplitudes and $\vartheta_{\text{unct},n}$ denoting the angle of the uncertainty component.

Furthermore, denote $r_n$, its LOS component and uncertainty component in a vector form by $\overrightarrow{\bf a}$, $\overrightarrow{\bf b}$ and $\overrightarrow{\bf c}$ respectively. As such, we have $\overrightarrow{\bf a}=\overrightarrow{\bf b}+\overrightarrow{\bf c}$, yielding a triangle. According to the property of triangles and noticing that $e_n$ is approximately the angle between $\overrightarrow{\bf a}$ and $\overrightarrow{\bf b}$, we have
$
\sin e_n =\frac{A_{\text{unct},n} }{A_{\text{ob},n}}\sin \theta_{\text{ob},n}
$,
where $\theta_{\text{ob},n} \triangleq   \theta_{\text{q},n}+ i_{N,n} \bar\theta_{\text {x-B2I}}+j_{N,n}  \bar\theta_{\text {y-B2I}}  - \vartheta_{\text{unct},n}$.

Using Taylor expansion, we obtain
$
e_n=\frac{A_{\text{unct},n} }{A_{\text{ob},n}}\sin \theta_{\text{ob},n}
$,
where we follow the fact $\underset{x \to 0} {\lim}\sin x=x$.

Due to a large Rician K-factor and a high received  SNR at the BS, $A_{\text{unct},n}$ is much smaller than $A_{\text{LOS},n}$. Hence, we have
\begin{align}
e_n \approx \frac{A_{\text{unct},n} }{A_{\text{LOS},n}}\sin \theta_{\text{ob},n}=
\frac{A_{\text{unct},n} }{\sqrt{\frac{\alpha_{\text{B2U}} v_{\text{B2U}} P_q}{v_{\text{B2U}}+1}}}\sin \theta_{\text{ob},n}
 .
\end{align}

Since $A_{\text{unct},n}$ and $\vartheta_{\text{unct},n}$ are the amplitude and the angle of a complex Gaussian random variable respectively, $A_{\text{unct},n}$ follows Rayleigh distribution and $\theta_{\text{ob},n}$ is uniformly distributed. Therefore, $e_n$ is circularly symmetric Gaussian with the variance  given by
\begin{align}
\sigma_e^2=\frac{1}{A_{\text{LOS},n}^2}
\mathbb{E}\left\{ {\left| A_{\text{unct},n} \right|}^2 \right\}
\mathbb{E}\left\{  {\left|  \sin \theta_{\text{ob},n}  \right|}^2 \right\}
\overset{(a)}{=} \frac{\left( 4-\pi \right)\left(v_{\text{B2U}}+1\right)}{8 \alpha_{\text{B2U}} P_q v_{\text{B2U}} } \left( \frac{\alpha_{\text{B2U}} P_q}{v_{\text{B2U}}+1}+\sigma_{\text{BS},0}^{2} \right),
\end{align}
where $(a)$ is derived according to $\mathbb{E}\left\{ {\left| A_{\text{unct},n} \right|}^2 \right\} = \frac{4-\pi}{4} \left( \frac{\alpha_\text{B2U} P_q}{v_\text{B2U}+1}+\sigma_{\text{BS},0}^{2} \right)$ and $\mathbb{E}\left\{  {\left|  \sin \theta_{\text{ob},n}  \right|}^2 \right\} =\frac{1}{2}$.

After some algebraic manipulations, we can obtained the desired result.

\section{Proof of Theorem \ref{t1}} \label{A2}
Define ${\bm{\Delta \vartheta}}\triangleq {\left[\Delta \bar\theta_{1,m_1},..., \Delta \bar\theta_{n,m_n}, ...,\Delta \bar\theta_{\frac{N}{2},m_{\frac{N}{2}}} \right]}^T$, and $\bm{\bar \theta}\triangleq {[\bar\theta_{\text{x-B2U} },\bar\theta_{\text {y-B2U}}]}^T$. The conditional probability density function (PDF) of ${\bm{\Delta \vartheta}}$  is given by
\begin{align}
p({\bm{\Delta \vartheta}}; \bm{\bar \theta})=
\prod_{n=1}^{\frac{N}{2}} \frac{1}{\sqrt{2 \pi \sigma_{\text{pd}}^2 }} \exp \left({-\frac{1}{2\sigma_{\text{pd}}^2} {\left(\Delta \bar\theta_{n,m_n}+\left( i_{N,n}-i_{N,m_n}\right)  \bar\theta_{\text {x-B2U}}+\left( j_{N,n}-j_{N,m_n}\right)
\bar\theta_{\text {y-B2U}} \right)}^2}\right).
\end{align}

According to the Neyman-Fisher Factorization theorem, we factor the above PDF as
\begin{align}
p({\bm{\Delta \vartheta}}; \bm{\bar \theta})=
f_1\left({\bm{\Delta \vartheta}} \right)
 f_2 \left(\bf{T}\left(\bm{\Delta \vartheta} \right),\bm{\bar \theta} \right),
\end{align}
where
\begin{align}
&f_1\left({\bm{\Delta \vartheta}} \right)=\frac{1}{{\left(\sqrt{2 \pi \sigma_{\text {pd}}}\right)}^{\frac{N}{2}} }\exp \left(- \frac{1}{2 \sigma_{\text{pd}}^{2} } \sum\limits_{n=1}^{\frac{N}{2}} \Delta \bar\theta_{n,m_n}^2 \right),\\
&f_2 \left(\bf{T}\left(\bm{\Delta \vartheta} \right),\bm{\bar \theta} \right)\!=\!\exp \left(\!- \frac{1}{2 \sigma_{\text{pd}}^{2} } \sum\limits_{n=1}^{\frac{N}{2}} {\left(\left( i_{N,n}-i_{N,m_n}\right) \bar\vartheta_{\text{x-B2U}}\!+\!\left( j_{N,n}-j_{N,m_n}\right) \bar\vartheta_{\text{y-B2U}} \right)}^2
\right)\\
&\times \exp \left(\!
\!-\frac{1}{ \sigma_{\text{pd}}^{2} } \left( \bar\vartheta_{\text{x-B2U}}  T_1\!+\! \bar\vartheta_{\text{y-B2U}} T_2\right)
\right)
,\nonumber \\
&{\bf{T}}\left(\bm{\Delta \vartheta} \right)=[T_1, T_2]^T.
\end{align}
where $T_1=\sum\limits_{n=1}^{\frac{N}{2}} \left( i_{N,n}-i_{N,m_n}\right) \Delta \bar\theta_{n,m_n}$ and $T_2=\sum\limits_{n=1}^{\frac{N}{2}} \left( j_{N,n}-j_{N,m_n}\right) \Delta \bar\theta_{n,m_n}$.

As such, we obtain a sufficient statistic ${\bf{T}}\left(\bm{\Delta \vartheta} \right)$ for $\bm{\bar \theta}$. Thus, $\hat{\bar{\bm \theta}}$ should be a  function with respect to ${\bf{T}}\left(\bm{\Delta \vartheta} \right)$. Also, noticing that $\hat{\bar{\bm \theta}}$ is unbiased, we obtain
$
\hat{\bar\theta}_{\text {x-B2U}}=-\frac{T_1}
{Q_1}$ and
$\hat{\bar\theta}_{\text {y-B2U}}=-\frac{ T_2}
{Q_1}
$,
where $Q_1 \triangleq \sum\limits_{n=1}^{N/2} {\left( i_{N,n}-i_{N,m_n}\right) }^2 =\frac{1}{6}N \left( N-1\right) $.

\section{Proof of Theorem \ref{t2}} \label{A3}
We first focus on the derivation of (\ref{E3}). Acording to Lemma \ref{L1}, we can express the  effective angle from the IRS to the user along $x$ axis as
\begin{align} \label{E7}
\bar\theta_{\text{x-I2U}}=\frac{\hat{d}_{\text{I2U}}}{d_{\text{I2U}}}\hat{\bar\theta}_{\text{x-I2U}}+\frac{d_{\text{B2U}}}{d_{\text{I2U}}} \epsilon_{\text{x-B2U}},
\end{align}
where
\begin{align}
d_{\text{I2U}}\!=\!\sqrt{ {\!\left(\!x_{\text{I}}\!-\!\hat{x}_{\text{U}}\!- \!\frac{d_{\text{B2U}}}{\pi}\epsilon_{\text{x-B2U}} \!\right)\!}^2
\!+\!{\!\left(\!y_{\text{I}}\!-\!\hat{y}_{\text{U}}\!- \!\frac{d_{\text{B2U}}}{\pi}\epsilon_{\text{y-B2U}} \!\right)\!}^2 \!+\!
{\left(z_{\text{I}}- z_{\text{U}}\right)}^2
}.
\end{align}

Recall ${z}_{\text{U}}=-\frac{ d_{\text{B2U}}  \sqrt{ \pi^2- {\bar \theta}_{\text{x-B2U}}^2- {\bar \theta}_{\text{y-B2U}}^2   }  }{\pi}$. Using the Taylor expansion of $(1+x)^{\frac{1}{2}}$ at $x=0$, we have
\begin{align}
{z}_{\text{U}} \!\approx \! d_{\text{B2U}}
\left(\! -1\!+\!\frac{{\bar \theta}_{\text{x-B2U}}^2\!+\!{\bar \theta}_{\text{y-B2U}}^2}{2 \pi^2}\right)\!=\!
\hat{z}_{\text{U}}\!+\!\frac{d_{\text{B2U}}
\left(\epsilon_{\text{x-B2U}}^2+\epsilon_{\text{y-B2U}}^2\!+\!
2\epsilon_{\text{x-B2U}}{\bar \theta}_{\text{x-B2U}}\!+\!2\epsilon_{\text{y-B2U}}{\bar \theta}_{\text{y-B2U}}\right)}{2 \pi^2}.
\end{align}

Then, we write $\frac{\hat{d}_{\text{I2U}}}{d_{\text{I2U}}}$ as
\begin{align} \label{E5}
\frac{\hat{d}_{\text{I2U}}}{d_{\text{I2U}}}=1/\sqrt{1+
Q_2 },
\end{align}
with
\begin{align}
Q_2
\approx 2\frac{
 \frac{  d_{\text{B2U}}}{\pi} \left({x}_{\text{I}}-\hat x_{\text{U}} \right)\epsilon_{\text{x-B2U}}
  +\frac{  d_{\text{B2U}}}{\pi} \left({y}_{\text{I}}-\hat y_{\text{U}} \right)\epsilon_{\text{y-B2U}}-
  \frac{ d_{\text{B2U}} }{\pi^2}  \left( z_{\text{I}}-\hat{z}_{\text U}\right)
  \left(\hat{\bar \theta}_{\text{x-I2U}}\epsilon_{\text{x-B2U}}+ \hat{\bar \theta}_{\text{y-I2U}}\epsilon_{\text{y-B2U}} \right) }
  {\hat{d}_{\text{I2U}}^2}. \nonumber
\end{align}

Recall that $\hat{\bar \theta}_{\text{x-I2U}}=\frac{\left({x}_{\text{I}}-\hat x_{\text{U}} \right)\pi}{\hat{d}_{\text{I2U}}}$ and $
\hat{\bar \theta}_{\text{y-I2U}}=\frac{\left({y}_{\text{I}}-\hat y_{\text{U}} \right)\pi}{\hat{d}_{\text{I2U}}}$. We can rewrite $Q_2$ as
\begin{align}
Q_2=
\frac{2 d_\text{B2U} }{\pi^2 \hat d_\text{I2U}}
\left\{
\left(\hat{\bar\theta}_{\text{x-I2U}}-\frac{\hat{\bar\theta}_{\text{x-I2U}} \hat{\bar\theta}_{\text{z-I2U}}}{\pi} \right)\epsilon_\text{x-B2U}
+
\left(\hat{\bar\theta}_{\text{y-I2U}}-\frac{\hat{\bar\theta}_{\text{y-I2U}} \hat{\bar\theta}_{\text{z-I2U}}}{\pi} \right)\epsilon_\text{y-B2U}
\right\},
\end{align}
where $\hat{\bar \theta}_{\text{z-I2U}} \triangleq \frac{\left({z}_{\text{I}}-\hat z_{\text{U}} \right)\pi}{\hat{d}_{\text{I2U}}}$.

Using the Taylor expansion of ${ \left(1+Q_2\right)}^{-\frac{1}{2}}$ at $ Q_2=0$, (\ref{E5}) can be approximated as
\begin{align} \label{E6}
\frac{\hat{d}_{\text{I2U}}}{d_{\text{I2U}}}
=1-\frac{1}{2}Q_2+{o}\left(Q_2 \right)
{\approx}1-\frac{1}{2}Q_2.
\end{align}

Similarly, we can express $\frac{d_{\text{B2U}}}{d_{\text{I2U}}} $ as
\begin{align} \label{E21}
\frac{d_{\text{B2U}}}{d_{\text{I2U}}} =\frac{d_{\text{B2U}}}{\hat{d}_{\text{I2U}}} \left( 1-\frac{1}{2}Q_2 \right).
\end{align}

Substituting (\ref{E6}) and (\ref{E21}) into (\ref{E7}), we have
\begin{align}
&\bar\theta_{\text{x-I2U}}= \left( 1-\frac{1}{2}Q_2 \right)\hat{\bar\theta}_{\text{x-I2U}}
+\frac{d_{\text{B2U}}}{\hat{d}_{\text{I2U}}} \left( 1-\frac{1}{2}Q_2 \right)\epsilon_{\text{x-B2U}}\\
&\overset{(a)}{\approx} \left( 1-\frac{1}{2}Q_2 \right)\hat{\bar\theta}_{\text{x-I2U}}+\frac{d_{\text{B2U}}}{\hat{d}_{\text{I2U}}}\epsilon_{\text{x-B2U}}.\nonumber
\end{align}

After some algebraic manipulations, we complete the proof of (\ref{E3}).
Following the similar process, we can obtain (\ref{E4}).

\section{Proof of Proposition \ref{p2}} \label{A4}
The average received power is  given by
\begin{align}
P_{\text{r}}= \mathbb{E}\left\{  {\left|{\bf g}^T {\bf w} \right|}^2 \right\}=\mathbb{E}\left\{  {\left| {\bf g}_{\text{LOS}}^T {\bf w} \right|}^2\right\}+\mathbb{E}\left\{  {\left| {\bf g}_{\text{NLOS}}^T {\bf w} \right|}^2\right\},
\end{align}
where the  effective channel is  defined as ${\bf g}^{T}={\bf h}_{\text{B2U}}^T  + {\bf h}_{\text{I2U}}^T {\bm \Theta} {\bf H}_{\text{B2I}}$
and decomposed  as ${\bf g}_{\text{LOS}}^T+{\bf g}_{\text{NLOS}}^T$ with
\begin{align}
&{\bf g}_{\text{LOS}}^T= \sqrt{\frac{\alpha_{\text{I2U}} \alpha_{\text{B2I}} v_{\text{B2I}} v_{\text{I2U}}}{\left(v_{\text{B2I}}+1\right) \left(v_{\text{I2U}}+1\right)}}
{\bf b}^T \left(\bar\theta_{\text{x-I2U}},\bar\theta_{\text{y-I2U}} \right) {\bf \Theta} {\bf \bar H}_{\text{B2I} }+ \sqrt{\frac{ \alpha_{\text{B2U}} v_{\text{B2U}}}{v_{\text{B2U}}+1}} {\bf a}^T \left(\bar\theta_{\text{x-B2U}},\bar\theta_{\text{y-B2U}} \right) ,\\
&{\bf g}_{\text{NLOS}}^T= \sqrt{\frac{\alpha_{\text{B2I}} \alpha_{\text{I2U}}}{v_{\text{B2I}}+1}} {\bf h}_{\text{I2U}}^{T} {\bf \Theta} {\bf \tilde H}_{\text{B2I} }
+\sqrt{\frac{\alpha_{\text{B2I}} \alpha_{\text{I2U}} v_{\text{B2I}}}{\left(v_{\text{B2I}}+1\right) \left(v_{\text{I2U}}+1\right)}} \tilde{\bf h}_{\text{I2U}}^{T} {\bf \Theta} {\bf \bar H}_{\text{B2I} }
+\sqrt{\frac{\alpha_{\text{B2U}}}{v_{\text{B2U}}+1}} {\bf \tilde h}_{\text{B2U}}^{T} ,
\end{align}
where $\bar{\bf H}_\text{B2I}\triangleq {\bf b} \left(\bar\theta_{\text{x-B2Ia}},\bar\theta_{\text{y-B2Ia}} \right)
{\bf a}^T \left(\bar \theta_{\text{x-B2I}},\bar\theta_{\text{y-B2I}} \right)$.

We start with the calculation of the second term:
\begin{align} \label{E11}
\mathbb{E}\left\{  {\left| {\bf g}_{\text{NLOS}}^T {\bf w} \right|}^2\right\}
=\text{tr}\left({\bf w}{\bf w}^H  \mathbb{E}\left\{ {\bf g}_{\text{NLOS}}^* {\bf g}_{\text{NLOS}}^T \right\} \right)
=\sigma_{\text{NLOS}}^2 {\bf w}^H  {\bf w},
\end{align}
where
\begin{align}
\sigma_{\text{NLOS}}^2 =M \frac{\alpha_{\text{I2U}} \alpha_{\text{B2I}} }{v_{\text{B2I}}+1}
\left( 1+\frac{v_{\text{B2I}}}{v_{\text{I2U}}+1}\right)
+\frac{\alpha_{\text{B2U}} }{v_{\text{B2U}}+1}.
\end{align}

Then, we calculate the first term:
\begin{align}
 \mathbb{E}\left\{  {\left| {\bf g}_{\text{LOS}}^T {\bf w} \right|}^2\right\}
 =T_1+T_2+2 \text{Re}\left( T_3 \right),
\end{align}
where
\begin{align}
&T_1=\mathbb{E}\left\{ {\bf b}^T\left(\bar\theta_{\text{x-I2U}},\bar\theta_{\text{y-I2U}} \right) {\bm \phi}_b  {\bm \phi}_b^H {\bf b}^*\left(\bar\theta_{\text{x-I2U}},\bar\theta_{\text{y-I2U}} \right) \right\},\\
&T_2=\mathbb{E}\left\{ {\bf a}^T\left(\bar\theta_{\text{x-B2U}},\bar\theta_{\text{y-B2U}} \right) {\bm \phi}_a  {\bm \phi}_a^H {\bf a}^*\left(\bar\theta_{\text{x-B2U}},\bar\theta_{\text{y-B2U}} \right) \right\},\\
&T_3=\mathbb{E}\left\{ {\bf a}^T\left(\bar\theta_{\text{x-B2U}},\bar\theta_{\text{y-B2U}} \right) {\bm \phi}_a  {\bm \phi}_b^H {\bf b}^*\left(\bar\theta_{\text{x-I2U}},\bar\theta_{\text{y-I2U}} \right) \right\},
\end{align}
with
\begin{align}
{\bm \phi}_a \triangleq \sqrt{\frac{v_{\text{B2U}} \alpha_{\text{B2U}} }{v_{\text{B2U}}+1}}{\bf w}, \
{\bm \phi}_b \triangleq \sqrt{\frac{v_{\text{B2I}} \alpha_{\text{B2I}} \alpha_{\text{I2U}} }{v_{\text{B2I}}+1}} {\bf \Theta} \bar{\bf H}_{\text{B2I}} {\bf w}.
\end{align}

 1) Calculate $T_1$
 \begin{align}
 T_1= \text{tr}\left( {\bm \phi}_b {\bm \phi}_b^H  {\bf B} \right)={\bm \phi}_b^H  {\bf B} {\bm \phi}_b ,
 \end{align}
where $ {\bf B} \triangleq \mathbb{E}\left\{ {\bf b}^*\left(\bar\theta_{\text{x-I2U}},\bar\theta_{\text{y-I2U}} \right) {\bf b}^T\left(\bar\theta_{\text{x-I2U}},\bar\theta_{\text{y-I2U}} \right) \right\}$ with  elements given by
\begin{align} \label{E8}
&{\left[{\bf B}\right]}_{mn}= \mathbb{E}\left\{ \exp\left( j \left(i_{M,n}-i_{M,m}\right) \bar\theta_{\text{x-I2U}}+ j \left(j_{M,n}-j_{M,m}\right) \bar\theta_{\text{y-I2U}} \right)  \right\}.
\end{align}

Substituting (\ref{E3}) and (\ref{E4}) into (\ref{E8}), we have
\begin{align} \label{E9}
&{\left[{\bf B}\right]}_{mn}
={\left[{\bf \hat B}\right]}_{mn} \mathbb{E}\left\{ \exp\left( j \left(i_{M,mn}\varphi_{1}+ j_{M,mn}\varphi_{2} \right) \epsilon_{\text{x-B2U}}+ j \left(i_{M,mn}\varphi_{2}+ j_{M,mn}\varphi_{\text{3}} \right) \epsilon_{\text{y-B2U}} \right)  \right\},
\end{align}
where  $ {\bf \hat B} \triangleq  {\bf b}^*\left(\hat{\bar\theta}_{\text{x-I2U}},\hat{\bar\theta}_{\text{y-I2U}} \right) {\bf b}^T\left(\hat{\bar\theta}_{\text{x-I2U}},\hat{\bar\theta}_{\text{y-I2U}} \right)  $, $i_{M,mn} \triangleq \left(i_{M,n}-i_{M,m}\right)$, $j_{M,mn} \triangleq \left(j_{M,n}-j_{M,m}\right)$.

Noticing that $\epsilon_{\text{x-B2U}}$ and $\epsilon_{\text{y-B2U}}$ follow complex Gaussian distribution $\mathcal{CN}\left(0,\sigma_{\text{est}}^2\right)$, (\ref{E9}) can be calculated as
\begin{align}
{\left[{\bf B}\right]}_{mn}={\left[{\bf \hat B}\right]}_{mn} \exp \left(- \frac{1}{2}\sigma_{\text{est}}^2\left\{ {\left(i_{M,mn}\varphi_{1}+ j_{M,mn}\varphi_{2} \right)}^2 +{\left(i_{M,mn}\varphi_{2}+ j_{M,mn}\varphi_{3} \right) }^2\right\} \right).
\end{align}

2) Calculate $T_2$
\begin{align}
 T_2= \text{tr}\left( {\bm \phi}_a {\bm \phi}_a^H  {\bf A} \right)={\bm \phi}_a^H  {\bf A} {\bm \phi}_a ,
 \end{align}
where $ {\bf A} \triangleq \mathbb{E}\left\{ {\bf a}^*\left(\bar\theta_{\text{x-B2U}},\bar\theta_{\text{y-B2U}} \right)  {\bf a}^T \left(\bar\theta_{\text{x-B2U}},\bar\theta_{\text{y-B2U}} \right) \right\}$
with elements given by
\begin{align}
{\left[{\bf A}\right]}_{mn}= \mathbb{E}\left\{ \exp\left( j i_{N,mn}\bar\theta_{\text{x-B2U}}+ j j_{N,mn} \bar\theta_{\text{y-B2U}} \right)  \right\}.
\end{align}

Invoking the results given by Theorem \ref{t2}, we have
\begin{align}
{\left[{\bf A}\right]}_{mn}={\left[{\bf \hat A}\right]}_{mn} \mathbb{E}\left\{ \exp\left( j i_{N,mn}\epsilon_{\text{x-B2U}}+ j j_{N,mn} \epsilon_{\text{y-B2U}} \right)  \right\},
\end{align}
where $ {\bf \hat A} \triangleq   {\bf a}^*\left(\hat{\bar\theta}_{\text{x-B2U}},\hat{\bar\theta}_{\text{y-B2U}} \right)  {\bf a}^T \left(\hat{\bar\theta}_{\text{x-B2U}},\hat{\bar\theta}_{\text{y-B2U}} \right) $.

Due to the fact that $\epsilon_{\text{x-B2U}}$ and $\epsilon_{\text{y-B2U}}$ follow complex Gaussian distribution $\mathcal{CN}\left(0,\sigma_{\text{est}}^2\right)$, ${\left[{\bf A}\right]}_{mn}$ can be computed as
\begin{align}
{\left[{\bf A}\right]}_{mn}={\left[{\bf \hat A}\right]}_{mn}
\exp \left(- \frac{1}{2}\sigma_{\text{est}}^2\left\{ i_{N,mn}^2 +j_{N,mn} ^2\right\} \right).
\end{align}

3) Calculate $T_3$
\begin{align}
T_3=\mathbb{E}\left\{ {\bf a}^T {\bm \phi}_a  {\bm \phi}_b^H {\bf b}^* \right\}
=\text{tr}\left( {\bm \phi}_a {\bm \phi}_b^H  {\bf C} \right)={\bm \phi}_b^H  {\bf C} {\bm \phi}_a ,
\end{align}
where ${\bf C} \triangleq \mathbb{E}\left\{ {\bf b}^*\left(\bar\theta_{\text{x-I2U}},\bar\theta_{\text{y-I2U}} \right)    {\bf a}^T \left(\bar\theta_{\text{x-B2U}},\bar\theta_{\text{y-B2U}} \right)   \right\}$ with elements ${\left[{\bf C}\right]}_{mn}$ given by
\begin{align}
{\left[{\bf C}\right]}_{mn}=\mathbb{E}\left\{ \exp\left( j  \left\{i_{M,m} \bar\theta_{\text{x-I2U}}
+j_{M,m} \bar\theta_{\text{y-I2U}} - i_{N,n} \bar\theta_{\text{x-B2U}}  -j_{N,n} \bar\theta_{\text{y-B2U}}
\right\} \right)  \right\}.
\end{align}

Invoking the results in Theorem \ref{t2}, we have
\begin{align}
&{\left[{\bf C}\right]}_{mn}
={\left[{\bf \hat C}\right]}_{mn} \\
&\times \mathbb{E}\left\{ \exp\left(  j\left(i_{M,m} \varphi_{1} +j_{M,m} \varphi_{2} -i_{N,n} \right)\epsilon_{\text{x-B2U}}
+ j\left( i_{M,m} \varphi_{2} +j_{M,m} \varphi_{3} -j_{N,n}\right)\epsilon_{\text{y-B2U}}  \right)  \right\},\nonumber
\end{align}
where ${\bf \hat C} \triangleq   {\bf b}^*\left(\hat{\bar\theta}_{\text{x-I2U}},\hat{\bar\theta}_{\text{y-I2U}} \right)    {\bf a}^T \left(\hat{\bar\theta}_{\text{x-B2U}},\hat{\bar\theta}_{\text{y-B2U}} \right)   $.

 According to the fact that $\epsilon_{\text{x-B2U}}$ and $\epsilon_{\text{y-B2U}}$ follow complex Gaussian distribution $\mathcal{CN}\left(0,\sigma_{\text{est}}^2\right)$, we have
 \begin{align}
&{\left[{\bf C}\right]}_{mn}
={\left[{\bf \hat C}\right]}_{mn}
\exp \left(-\frac{1}{2}\sigma_{\text{est}}^2\left\{ {\left(i_{M,m} \varphi_{1} +j_{M,m} \varphi_{2} -i_{N,n} \right)}^2
+{\left(i_{M,m} \varphi_{2} +j_{M,m} \varphi_{3} -j_{N,n}\right)}^2\right\} \right).
\end{align}

Combining 1), 2) and 3), we have
\begin{align} \label{E10}
\mathbb{E}\left\{  {\left| {\bf g}_{\text{LOS}}^T {\bf w} \right|}^2\right\}
 ={\bm \phi}_b^H  {\bf B} {\bm \phi}_b+{\bm \phi}_a^H  {\bf A} {\bm \phi}_a+2 \text{Re}\left( {\bm \phi}_b^H  {\bf C} {\bm \phi}_a \right),
\end{align}

Combining (\ref{E10}) and (\ref{E11})  yields the desired result.

\section{Proof of Proposition \ref{p5}}\label{A7}
The constant-modulus constraint $|\xi_i|=1,i=1,\dots,M$ is equivalent to the following two constraints:
\begin{align}
    \text{tr}\left( {\bm \xi}\right)=M,\\
    \| {\bm \xi}_\infty\|\le 1.
\end{align}
Since we project ${\bf g}_{\text{gd}}=-\nabla_{{\bm \xi}} G{\left(\bm{\xi}  \right)}$ into the tangent plane of $\text{tr}\left({\bm \xi}{\bm \xi}^H \right)=M$:
$
  {\bf g}_{\text{p}}={\bf g}_{\text{gd}}-\frac{{\bf g}_{\text{gd}}^T{\bm \xi}^* {\bm \xi}}{{\left\| {\bm \xi} \right\|}^2},
$
and use ${\bf g}_{\text{p}}$ as the search direction,
the first constraint $\text{tr}\left( {\bm \xi}\right)=M$ holds.
Noticing that $\underset{p\to \infty}{\lim} \ell_{p}=\ell_{\infty}$, the constraint $\| {\bm \xi}_\infty\|\le 1$ is approximately equivalent to $\| {\bm \xi}_p\|\le 1$ with a large $p$. Then a barrier method is exploited to make  $\| {\bm \xi}_p\|\le 1$ satisfied.
To this end, we complete our proof.

\section{Proof of Corollary \ref{c2}}\label{A5}

Recall that the optimal ${\bf w}$ is  $\sqrt{P_{\text{BS}}} {\bf t}_{\text{max}}$  with   ${\bf t}_{\text{max}}$ being the eigenvector of ${\bf T}$ corresponding to the largest eigenvalue $\lambda_{\text{max}}$.
Thus, the average received signal power is given by
$
P_{\text{r}}=P_{\text{BS}} \lambda_{\text{max}}.
$

Due to the fact that the trace of a matrix is the sum of all eigenvalues, the following equation holds
$
 \text{tr}\left( {\bf T}\right)= \sum\limits_{n=1}^{N} \lambda_{T,n}.
$
Since the channel is sparse, the largest eigenvalue $\lambda_{\text{max}}$ dominates the trace $\text{tr}\left( {\bf T}\right)$. As such,  we have $\lambda_{\text{max}} \approx \text{tr}\left( {\bf T}\right)$.

Thus, the achievable rate can be approximated as
\begin{align}
R_{\text{approx}}=\log_2 \left(1+\frac{ P_{\text{BS}} \text{tr}\left( {\bf T}\right)  }{\sigma_0^2} \right).
\end{align}

Starting from  Theorem \ref{t3}, $\text{tr}\left( {\bf T}\right)$ can be expressed as
\begin{align}
\text{tr}\left( {\bf T}\right)&=
\beta_{\text{B2I2U}} \sum\limits_{m=1}^{M} \sum\limits_{n=1}^{M} { \xi}_{m}^{*}{ \xi}_{n} {\left[ \bf B \right]}_{mn}
\left( \sum\limits_{i=1}^{N}  {\left[ \bar{\bf H}_{\text{B2I}}^* \right]}_{mi}  {\left[ \bar{\bf H}_{\text{B2I}} \right]}_{ni} \right) \\
&+2\sqrt{\beta_{\text{B2I2U}} \beta_{\text{B2U}} } \text{Re} \left\{ \sum\limits_{m=1}^{M} { \xi}_{m}^{*}
 \left(\sum\limits_{i=1}^{N} {\left[ \bf C \right]}_{mi}  {\left[ \bar{\bf H}_{\text{B2I}}^* \right]}_{mi}\right)\right\} +N \beta_{\text{B2U}}+N \sigma_{\text{NLOS}}^2. \nonumber
\end{align}

Recall that $\bar{\bf H}_{\text{B2I}} ={\bf b} \left(\bar\theta_{\text{x-B2Ia}},\bar\theta_{\text{y-B2Ia}} \right)
 {\bf a}^T \left(\bar \theta_{\text{x-B2I}},\bar\theta_{\text{y-B2I}} \right)$. We can obtain the desired result.
 \end{appendices}

\bibliographystyle{IEEEtran}
\bibliography{main}{}

\end{document}